\documentclass[12pt]{article}

\usepackage{amssymb}
\usepackage{fullpage}
\usepackage{graphicx}

\newtheorem{theorem}{Theorem}
\newtheorem{lemma}{Lemma}

\newcommand{\IR}{\mathbb{R}}
\newcommand{\IS}{\mathbb{S}}

\newcommand{\skel}{\mathord{\it skel}}
\newcommand{\Ann}{\mathord{\it Ann}}
\newcommand{\Shell}{\mathord{\it Shell}}
\newcommand{\Dil}{\mathord{\it Dil}}

\newcommand{\qed}{\rule{0.5em}{1.5ex}}
\newcommand{\fqed}{{\hfill~\qed}}
\newenvironment{proof}{{\noindent \bf Proof.}}
                      {{\hfill \fqed} \vspace{1em}}

\title{On the Stretch Factor of Convex Polyhedra whose Vertices are 
      (Almost) on a Sphere}
\author{
Prosenjit Bose\thanks{School of Computer Science, Carleton University, 
    Ottawa, Canada. These authors were supported by the 
    Natural Sciences and Engineering Research Council of Canada. 
    D.H.\ was supported by an Ontario Graduate Scholarship.}
\and 
Paz Carmi\thanks{Department of Computer Science, Ben-Gurion University 
   of the Negev, Israel.} 
\and 
Mirela Damian\thanks{Department of Computer Science, Villanova 
   University, Villanova, PA 19403, USA. Supported by NSF grant 
   CCF-1218814.}
\and 
Jean-Lou De Carufel\thanks{School of Electrical Engineering and 
    Computer Science, University of Ottawa, Canada.} 
\and  
Darryl Hill\footnotemark[1] 
\and 
Anil Maheshwari\footnotemark[1] 
\and
Yuyang Liu\footnotemark[1] 
\and 
Michiel Smid\footnotemark[1] 
} 
\date{\today}

\begin{document} 

\maketitle 

\begin{abstract} 
Let $P$ be a convex simplicial polyhedron in $\IR^3$. The skeleton of 
$P$ is the graph whose vertices and edges are the vertices and edges 
of $P$, respectively. We prove that, if these vertices are on a 
sphere, the skeleton is a $(0.999 \cdot \pi)$-spanner. If the 
vertices are very close to a sphere, then the skeleton is not 
necessarily a spanner. For the case when the boundary of $P$ is between 
two concentric spheres of radii $r$ and $R$, where $R>r>0$, and the 
angles in all faces are at least $\theta$, we prove that the skeleton 
is a $t$-spanner, where $t$ depends only on $R/r$ and $\theta$.  
One of the ingredients in the proof is a tight upper bound on 
the geometric dilation of a convex cycle that is contained in an 
annulus. 
\end{abstract}

\section{Introduction}
Let $S$ be a finite set of points in Euclidean space and let $G$ be a 
graph with vertex set $S$. We denote the Euclidean distance between any 
two points $p$ and $q$ by $|pq|$. Let the length of any edge $pq$ in $G$ 
be equal to $|pq|$, and define the length of a path in $G$ to be the sum 
of the lengths of the edges on this path. For any two vertices $p$ and 
$q$ in $G$, we denote by $|pq|_G$ the length of a shortest path in $G$ 
between $p$ and $q$. For a real number $t \geq 1$, we say that $G$ is a 
\emph{$t$-spanner} of $S$, if $|pq|_G \leq t |pq|$ for all vertices $p$ 
and $q$. The \emph{stretch factor} of $G$ is the smallest value of $t$ 
such that $G$ is a Euclidean $t$-spanner of $S$. See~\cite{ns-gsn-07} 
for an overview of results on Euclidean spanners.

It is well-known that the stretch factor of the Delaunay triangulation
in $\IR^2$ is bounded from above by a constant. The first proof of this
fact is due to Dobkin \emph{et al.}~\cite{dfs-dgaag-90}, who obtained
an upper bound of $(1+\sqrt{5})\pi/2 \approx 5.08$. The currently best
known upper bound, due to Xia~\cite{x-sfdtl-13}, is $1.998$.

Let $P$ be a convex simplicial polyhedron in $\IR^3$, i.e., all faces
of $P$ are triangles. The \emph{skeleton} of $P$, denoted by $\skel(P)$, 
is the graph whose vertex and edge sets are equal to the vertex and edge 
sets of $P$. 

Since there is a close connection between Delaunay triangulations in 
$\IR^2$ and convex hulls in $\IR^3$, it is natural to ask if the skeleton 
of a convex simplicial polyhedron in $\IR^3$ has a bounded stretch 
factor. By taking a long and skinny convex polyhedron, however, this is 
clearly not the case. 

In 1987, Raghavan suggested, in a private communication to
Dobkin \emph{et al.}~\cite{dfs-dgaag-90}, that the skeleton of a convex 
simplicial polyhedron, all of whose vertices are on a sphere, has 
bounded stretch factor. Consider such a polyhedron $P$. By a 
translation and scaling, we may assume that the vertex set $S$ of $P$ 
is on the unit-sphere 
\[ \IS^2 = \{ (x,y,z) \in \IR^3 : x^2 + y^2 + z^2 = 1 \} .  
\] 
Observe that this does not change the stretch factor of $P$'s skeleton.  
It is well-known that the convex hull of $S$ (i.e., the polyhedron $P$) 
has the same combinatorial structure as the spherical Delaunay 
triangulation of $S$; this was first  observed by 
Brown~\cite{b-gtfga-80}. Based on this, 
Bose \emph{et al.}~\cite{bps-chpss-14} showed that the proof of   
Dobkin \emph{et al.}~\cite{dfs-dgaag-90} can be modified to prove that 
the skeleton of $P$ is a $t$-spanner of its vertex set $S$, where 
$t= \frac{3 \pi}{2} (1+\pi/2) \approx 12.115$.  

In Section~\ref{seconsphere}, we improve the upper bound on the stretch 
factor to $0.999 \cdot \pi \approx 3.138$. Our proof considers any two 
vertices $p$ and $q$ of $P$ and the plane $H_{pq}$ through $p$, $q$, 
and the origin. The great arc on $\IS^2$ connecting $p$ and $q$ is 
contained in $H_{pq}$. The path on the convex polygon 
$Q_{pq} = P \cap H_{pq}$ that is on the same side of $pq$ as this 
great arc passes through a sequence of triangular faces of $P$. 
An edge-unfolding of these faces results in a sequence of triangles in 
a plane, whose circumdisks form a \emph{chain of disks}, as defined 
by Xia~\cite{x-sfdtl-13}. The results of Xia then imply the upper 
bound of $0.999 \cdot \pi$ on the stretch factor of the skeleton of $P$. 

A natural question is whether a similar result holds for a convex 
simplicial polyhedron whose vertices are ``almost'' on a sphere. 
In Section~\ref{secalmost}, we show that this is not the case: We give 
an example of a set of points that are very close to a sphere, such that 
the skeleton of their convex hull has an unbounded stretch factor. 

In Section~\ref{secACSH}, we consider convex simplicial polyhedra $P$ 
whose boundaries are between two concentric spheres of radii $r$ and 
$R$, where $R>r>0$, that contain the common center of these spheres, 
and in which the angles in all faces are at least $\theta$. 
We may assume that the two spheres are centered at the origin. We 
present an improvement of a result by 
Karavelas and Guibas~\cite{kg-skgsa-01}, i.e., we show that for any two 
vertices $p$ and $q$, their shortest-path distance in the skeleton of 
$P$ is at most $(1 + 1/\sin(\theta/2))/2$ times their shortest-path 
distance along the surface of $P$. The latter shortest-path distance is 
at most the shortest-path distance between $p$ and $q$ along the 
boundary of the convex polygon $Q_{pq}$ which is obtained by 
intersecting $P$ with the plane through $p$, $q$, and the origin. This 
convex polygon contains the origin and its boundary is contained 
between the two circles of radii $r$ and $R$ that are centered at the 
origin. Gr{\"u}ne~\cite[Lemma~2.40]{g-gdhd-06} has 
shown that the stretch factor of any such polygon is at most 
\[ \frac{\pi R/r}{2 - (\pi/2) (R/r-1)} ,  
\] 
provided that $R/r < 1 + 4/\pi$. In Section~\ref{secannulus}, we improve 
this upper bound to 
\[ \sqrt{(R/r)^2-1} + (R/r) \arcsin(r/R) , 
\]
which is valid, and tight, for all $R>r>0$. As a result, the stretch 
factor of the skeleton of $P$ is at most 
\[ \frac{1 + 1/\sin(\theta/2)}{2}   
   \left( \sqrt{(R/r)^2-1} + (R/r) \arcsin(r/R) \right) .
\]

\section{Convex Polyhedra whose Vertices are on a Sphere}  
\label{seconsphere} 
In this section, we prove an upper bound on the stretch factor of the 
skeleton of a convex simplicial polyhedron whose vertices are on a 
sphere. As we will see in Section~\ref{subsecSFCP}, our 
upper bound follows from Xia's upper bound in~\cite{x-sfdtl-13} on the 
stretch factor of chains of disks in $\IR^2$. We start by reviewing 
such chains. 

\subsection{Chains of Disks} \label{secCoD}
Let $\mathcal{D} = (D_1,D_2,\ldots,D_k)$ be a sequence of disks in 
$\IR^2$, where $k \geq 2$. For each $i$ with $2 \leq i \leq k$, define
\[ C_i^{i-1} = D_{i-1} \cap \partial D_i ,
\]
i.e., $C_i^{i-1}$ is that part of the boundary of $D_i$ that is contained 
in $D_{i-1}$. Similarly, for each $i$ with $1 \leq i < k$, define 
\[ C_i^{i+1} = D_{i+1} \cap \partial D_i .
\] 
The sequence $\mathcal{D}$ of disks is called a \emph{chain of disks}, 
if  
\begin{enumerate} 
\item for each $i$ with $1 \leq i < k$, the circles $\partial D_i$ and 
      $\partial D_{i+1}$ intersect in exactly one or two points, and  
\item for each $i$ with $2 \leq i < k$, the circular arcs $C_i^{i-1}$ 
      and $C_i^{i+1}$ have at most one point in common. 
\end{enumerate} 
See Figure~\ref{figCoD} for an example. 

\begin{figure}
\begin{center}
\includegraphics[scale=0.7]{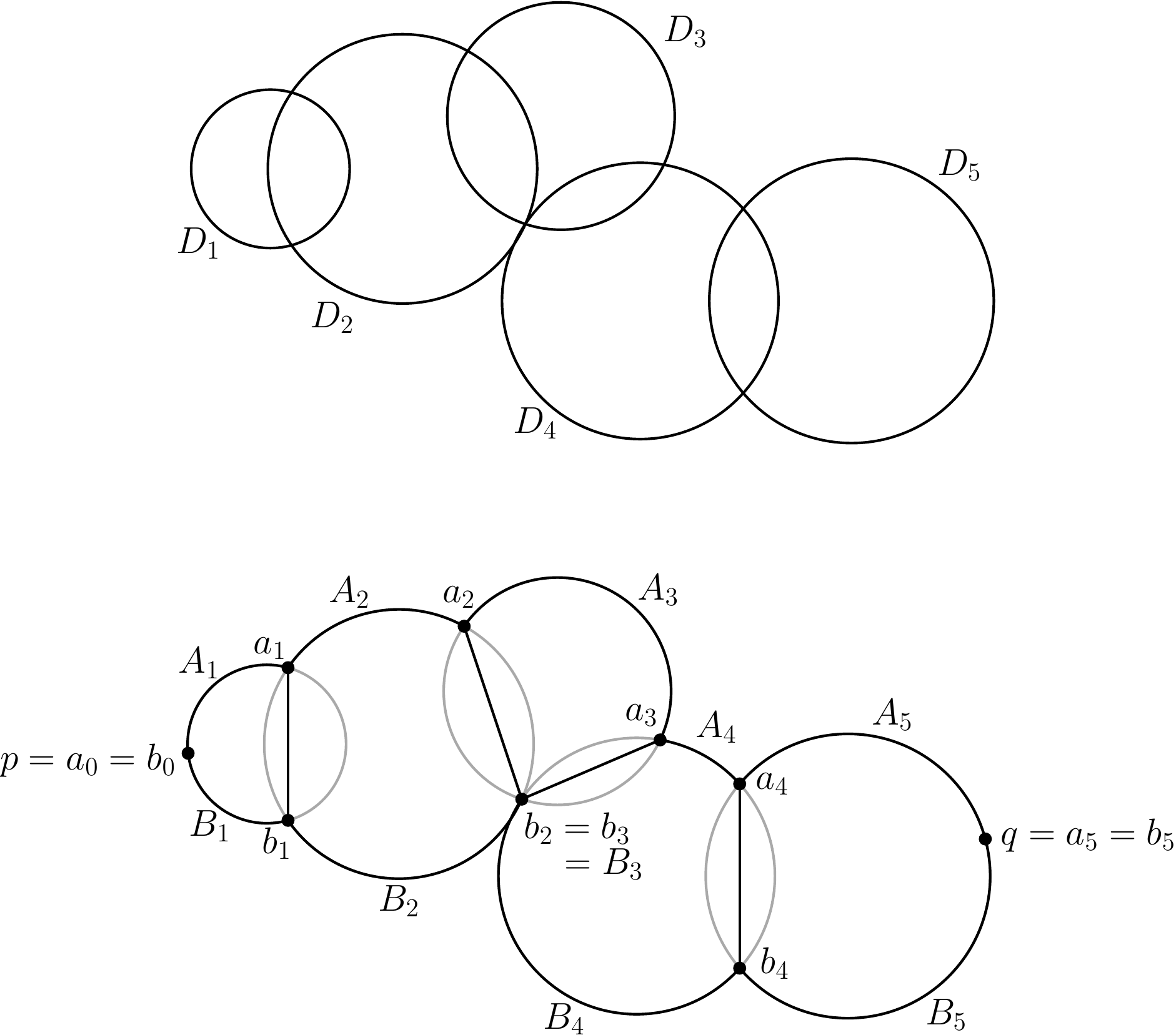}
\end{center}
\caption{The top figure shows a chain 
    $\mathcal{D} = (D_1,D_2,\ldots,D_5)$ of disks. The bottom figure 
    shows the graph $G(\mathcal{D},p,q)$; the edges of this graph are 
    black. The edge $B_3$ has length zero; it consists of just the 
    point $b_2$ (which is equal to $b_3$).}
\label{figCoD}
\end{figure}

Let $p$ and $q$ be two distinct points in the plane such that
\begin{enumerate} 
\item $p$ is on $\partial D_1$ and not in the interior of $D_2$, and  
\item $q$ is on $\partial D_k$ and not in the interior of $D_{k-1}$.  
\end{enumerate} 
For each $i$ with $1 \leq i < k$, let $a_i$ and $b_i$ be the intersection 
points of the circles $\partial D_i$ and $\partial D_{i+1}$, where 
$a_i = b_i$ if these two circles are tangent. We label these intersection 
points in such a way that $a_i$ is on or to the left of the directed 
line from the center of $D_i$ to the center of $D_{i+1}$, and $b_i$ is on 
or to the right of this line. Define $a_0 = p$, $b_0 = p$, $a_k = q$, and 
$b_k = q$. For each $i$ with $1 \leq i \leq k$, let $A_i$ be the circular 
arc on $\partial D_i$ connecting the points $a_{i-1}$ and $a_i$ that 
is completely on the same side of $\Pi$ as $a_{i-1}$ and $a_i$, and 
let $B_i$ be the circular arc on $\partial D_i$ connecting the points 
$b_{i-1}$ and $b_i$ that is on the same side of $\Pi$ as $b_{i-1}$ and 
$b_i$. 

Consider the graph $G(\mathcal{D},p,q)$ with vertex set 
$\{p,a_1,a_2,\ldots,a_{k-1},b_1,b_2,\ldots,b_{k-1},q\}$ and edge set 
consisting of 
\begin{itemize}
\item the circular arcs $A_1,A_2,\ldots,A_k$, 
\item the circular arcs $B_1,B_2,\ldots,B_k$, and 
\item the line segments $a_1 b_1$, $a_2 b_2$, \ldots , $a_{k-1} b_{k-1}$.
\end{itemize} 
Figure~\ref{figCoD} shows an example. 
 
For each $i$ with $1 \leq i \leq k$, the lengths of the edges $A_i$ and 
$B_i$ are equal to the lengths $|A_i|$ and $|B_i|$ of these arcs,
respectively. For each $i$ with $1 \leq i < k$, the length of the edge 
$a_i b_i$ is equal to $|a_i b_i|$. The length of a shortest path in 
$G(\mathcal{D},p,q)$ is denoted by $|pq|_{G(\mathcal{D},p,q)}$. 

\begin{theorem}[Xia~\cite{x-sfdtl-13}]   \label{thmxia} 
Let $L$ be the length of any polygonal path that starts at $p$, ends at 
$q$, and intersects the line segments 
$a_1 b_1$, $a_2 b_2$, \ldots , $a_{k-1} b_{k-1}$ in this order. Then, 
\[ |pq|_{G(\mathcal{D},p,q)} \leq 1.998 \cdot L . 
\]  
\end{theorem}

\subsection{Bounding the Stretch Factor}  \label{subsecSFCP} 
Let $P$ be a convex simplicial polyhedron in $\IR^3$ and assume that 
all vertices of $P$ are on a sphere. By a translation and scaling, we 
may assume that this sphere is the unit-sphere $\IS^2$ (without changing 
the stretch factor of $P$'s skeleton). We assume that 
(i) no four vertices of $P$ are co-planar and (ii) the plane through 
any three vertices of $P$ does not contain the origin. 

Fix two distinct vertices $p$ and $q$ of $P$. We will prove that 
$|pq|_{\skel(P)}$, i.e., the length of a shortest path in the skeleton 
$\skel(P)$ of $P$, is at most $0.999 \cdot \pi \cdot |pq|$. If $pq$ is 
an edge of $\skel(P)$, then this claim obviously holds. We assume from 
now on that $pq$ is not an edge of $\skel(P)$. 

Our proof will use the following notation (refer to 
Figure~\ref{figdefs}): 

\begin{figure}
\begin{center}
\includegraphics[scale=0.6]{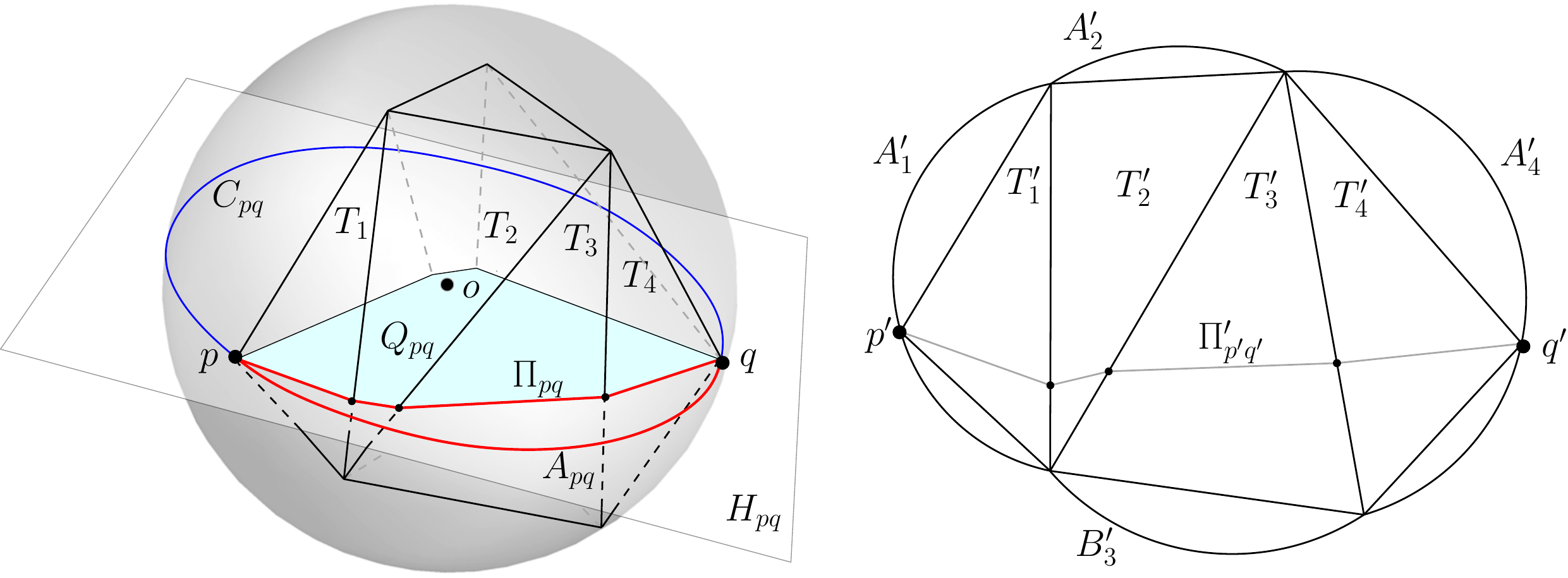}
\end{center}
\caption{Illustrating the notation in Section~\ref{subsecSFCP}.}  
\label{figdefs}
\end{figure}

\begin{itemize} 
\item $H_{pq}$: the plane through $p$, $q$, and the origin (i.e., the 
      center of $\IS^2$). 
\item $C_{pq}$: the circle $\IS^2 \cap H_{pq}$. 
\item $A_{pq}$: the shorter arc of $C_{pq}$ connecting $p$ and $q$. 
\item $Q_{pq}$: the convex polygon $P \cap H_{pq}$. 
\item $\Pi_{pq}$: the path along $Q_{pq}$ from $p$ to $q$ that is on 
      the same side of the line segment $pq$ as the arc $A_{pq}$; 
      observe that $\Pi_{pq}$ is a path between $p$ and $q$ along the 
      surface of $P$. 
\item $T_1,T_2,\ldots,T_k$: the sequence of faces of $P$ that the path 
      $\Pi_{pq}$ passes through. Observe that $k \geq 2$.
\end{itemize} 
Let $T'_1,T'_2,\ldots,T'_k$ be the sequence of triangles obtained from 
an edge-unfolding of the triangles $T_1,T_2,\ldots,T_k$. Thus, 
\begin{itemize}
\item all triangles $T'_1,T'_2,\ldots,T'_k$ are contained in one plane, 
\item for each $i$ with $1 \leq i \leq k$, the triangles $T_i$ and $T'_i$ 
      are congruent, and 
\item for each $i$ with $1 \leq i < k$, the triangles $T'_i$ and 
      $T'_{i+1}$ share an edge, which is the ``same'' edge that is shared 
      by $T_i$ and $T_{i+1}$, and the interiors of $T'_i$ and $T'_{i+1}$ 
      are disjoint.  
\end{itemize} 
For each $i$ with $1 \leq i \leq k$, let $D'_i$ be the circumdisk of the 
triangle $T'_i$. Let $\mathcal{D'} = (D'_1,D'_2,\ldots,D'_k)$ and let 
$p'$ and $q'$ be the vertices of $T'_1$ and $T'_k$ corresponding to $p$ 
and $q$, respectively. We will prove the following lemma in 
Section~\ref{seclemCoD}.  

\begin{lemma}  \label{lemCoD}  
The following properties hold: 
\begin{enumerate} 
\item The sequence $\mathcal{D'}$ is a chain of disks. 
\item $p'$ is on $\partial D'_1$ and not in the interior of $D'_2$. 
\item $q'$ is on $\partial D'_k$ and not in the interior of $D'_{k-1}$.  
\end{enumerate} 
\end{lemma} 

Consider the graph $G(\mathcal{D'},p',q')$ that is defined by 
$\mathcal{D'}$ and the two points $p'$ and $q'$; see 
Section~\ref{secCoD}. We first observe that $|pq|_{\skel(P)}$ is at most 
the shortest-path distance between $p$ and $q$ in the graph consisting 
of all vertices and edges of the faces $T_1,T_2,\ldots,T_k$. The latter 
shortest-path distance is equal to the shortest-path distance between 
$p'$ and $q'$ in the graph consisting of all vertices and edges of the 
triangles $T'_1,T'_2,\ldots,T'_k$. Since the latter shortest-path 
distance is at most $|p'q'|_{G(\mathcal{D'},p',q')}$, it follows that 
\[ |pq|_{\skel(P)} \leq |p'q'|_{G(\mathcal{D'},p',q')} . 
\] 
Let $\Pi'_{p'q'}$ be the path through $T'_1,T'_2,\ldots,T'_k$ 
corresponding to the path $\Pi_{pq}$. By Lemma~\ref{lemCoD} and 
Theorem~\ref{thmxia}, we have 
\[ |p'q'|_{G(\mathcal{D'},p',q')} \leq 1.998 \cdot |\Pi'_{p'q'}| . 
\]
Since $|\Pi'_{p'q'}| = |\Pi_{pq}|$, it follows that 
\[ |pq|_{\skel(P)} \leq 1.998 \cdot |\Pi_{pq}| . 
\] 

It remains to bound $|\Pi_{pq}|$ in terms of the Euclidean distance 
$|pq|$. Consider again the plane $H_{pq}$ through $p$, $q$, and the 
origin, the circle $C_{pq} = \IS^2 \cap H_{pq}$, the shorter arc 
$A_{pq}$ of $C_{pq}$ connecting $p$ and $q$, and the convex polygon 
$Q_{pq} = P \cap H_{pq}$. Observe that both $p$ and $q$ are on $C_{pq}$, 
and both these points are vertices of $Q_{pq}$. Moreover, $Q_{pq}$ is 
contained in the disk with boundary $C_{pq}$. It follows that 
\[ |\Pi_{pq}| \leq |A_{pq}| . 
\] 

Let $\alpha$ be the angle between the two vectors from the 
origin (which is the center of $C_{pq}$) to $p$ and $q$; see the figure 
below. 

\begin{center}
\includegraphics[scale=0.60]{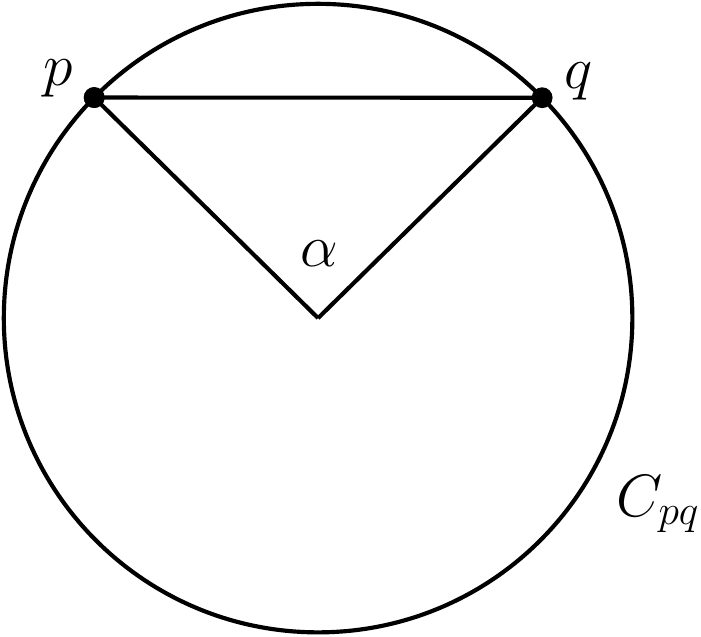}
\end{center}

Since $C_{pq}$ has radius $1$, we have $|A_{pq}| = \alpha$ and 
$|pq| = 2 \sin (\alpha/2)$. Therefore, 
\[ |A_{pq}| = \frac{\alpha/2}{\sin (\alpha/2)} \cdot |pq| .
\]
The function $g(x) = x/\sin x$ is increasing for $0 \leq x \leq \pi/2$ 
(because its derivative is positive for $0 < x \leq \pi/2$), implying 
that
\[ |A_{pq}| \leq g(\pi/2) \cdot |pq| = (\pi/2) \cdot |pq| .
\]
By combining the above inequalities, we obtain 
\[ |pq|_{\skel(P)} \leq 1.998 \cdot \pi/2 \cdot |pq| . 
\]
Thus, assuming Lemma~\ref{lemCoD} holds, we have proved the following 
result. 

\begin{theorem}   \label{thm999}  
Let $P$ be a convex simplicial polyhedron in $\IR^3$, all of whose 
vertices are on a sphere. Assume that no four vertices of $P$ are 
co-planar and the plane through any three vertices of $P$ does not 
contain the center of the sphere. Then the skeleton of $P$ is a 
$t$-spanner of the vertex set of $P$, where 
\[ t = 0.999 \cdot \pi .
\] 
\end{theorem}

\subsection{Proof of Lemma~\ref{lemCoD}}  \label{seclemCoD}  
Lemma~\ref{lemCoD} will follow from Lemma~\ref{lemlocallyD} below.  
The proof of the latter lemma uses an additional result: 

\begin{lemma}   \label{lembelow}
Let $i$ be an integer with $1 \leq i \leq k$. The polyhedron $P$ and the 
origin are in the same closed halfspace that is bounded by the plane 
through the face $T_i$ of $P$.  
\end{lemma}
\begin{proof} 
Let $e_i$ be the edge of the convex polygon $Q_{pq}$ that spans the 
face $T_i$. Since the path $\Pi_{pq}$ (which contains $e_i$ as an edge) 
is on the same side of the line segment $pq$ as the arc $A_{pq}$, and 
since the origin is on the other side of this line segment, the polygon 
$Q_{pq}$ and the origin are in the same closed halfplane (in $H_{pq}$) 
that is bounded by the line through $e_i$. This implies the claim. 
\end{proof} 

\begin{lemma}   \label{lemlocallyD} 
Let $i$ be an integer with $1 \leq i < k$ and let $w$ be the vertex of 
$T_{i+1}$ that is not a vertex of $T_i$. Consider the vertex $w'$ of the 
unfolded triangle $T'_{i+1}$ that corresponds to $w$. Then $w'$ is not in 
the circumdisk $D'_i$ of the unfolded triangle $T'_i$.  
\end{lemma} 
\begin{proof}  
Let $T_i = \triangle uvq$ and $T_{i+1} = \triangle uvw$; thus, $uv$ is 
the edge shared by the faces $T_i$ and $T_{i+1}$ of $P$. Assume without 
loss of generality that $uv$ is parallel to the $z$-axis. Let $C$ be 
the cross-section of $\IS^2$ that passes through $w$ and is orthogonal 
to $uv$. Let $u'$, $v'$, $q'$, and $o'$ be the orthogonal projections 
of $u$, $v$, $q$, and $o$ onto the plane supporting $C$, respectively; 
refer to Figure~\ref{fig:lem3}.

\begin{figure}
\centering
\includegraphics[width=0.7\linewidth]{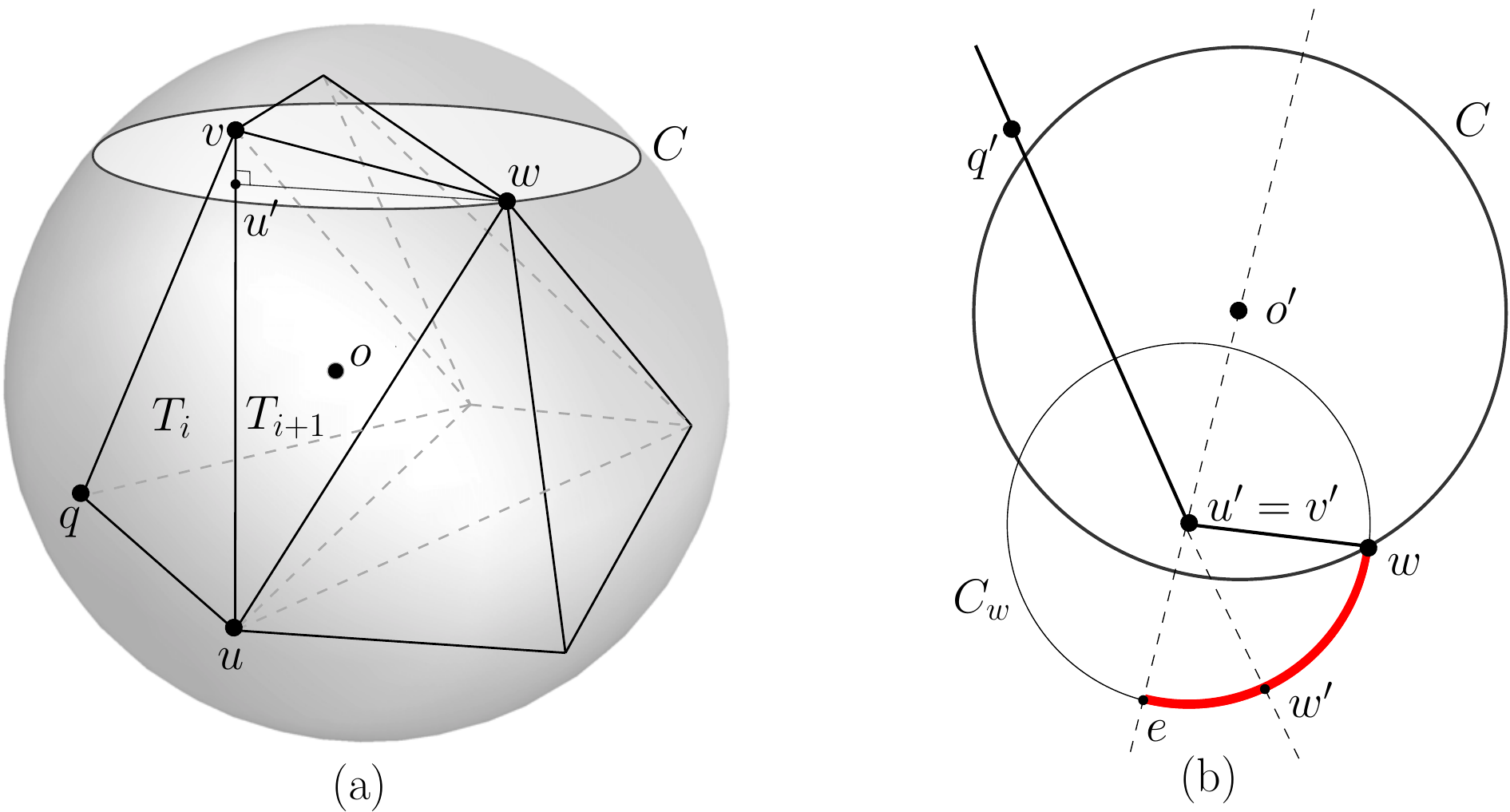}
\caption{Cross-section through $w$ orthogonal to $uv$ (a) side view 
of $C$ (b) top view of $C$. }
\label{fig:lem3}
\end{figure}

Let $C_w$ be the circle with center $u'$ and radius $|u'w|$ that is 
coplanar with $C$. Since $T'_{i+1}$ is obtained by rotating $T_{i+1}$ 
about the line through $u$ and $v$, it must be that $w'$ lies on 
$\partial C_w$. Next we show that $w'$ is exterior to $C$.

Let $e$ be the intersection point between the line supporting $o'u'$ 
and $C_w$ that is farthest away from $o'$. By Lemma~\ref{lembelow}, 
$o'$ lies interior to the convex angle $\angle{wu'q'}$ and, therefore, 
$e$ lies exterior to $C$. Moreover, the circular arc of $\partial C_w$ 
with endpoints $w$ and $e$ that extends from $w$ away from $o'$ 
(marked as a thick curve in Figure~\ref{fig:lem3}b), lies exterior of 
$C$. This circular arc is precisely the locus of $w'$. It follows that 
$w'$ is exterior to $C$, which implies that $w'$ is exterior to 
$\IS^2$. Now observe that $D'_i$ is congruent with the disk $D_i$ 
obtained by intersecting $\IS^2$ with the plane supporting $T_i$. 
Since $w'$ is coplanar with $D_i$ and exterior to $D_i$, we have 
that $w'$ is exterior to $D'_i$. This concludes the proof.
\end{proof} 

It is easy to see that Lemma~\ref{lemlocallyD} implies that the sequence 
$\mathcal{D'}$ is a chain of disks, i.e., this sequence satisfies the 
two properties given in Section~\ref{secCoD}. Moreover, it follows from 
Lemma~\ref{lemlocallyD} that $p'$ is on $\partial D'_1$ and not in the 
interior of $D'_2$ and $q'$ is on $\partial D'_k$ and not in the interior 
of $D'_{k-1}$. Thus, we have completed the proof of Lemma~\ref{lemCoD}.

\subsection{Convex Polyhedra whose Vertices are Almost on a Sphere}  
\label{secalmost} 

In this section, we give an example of a convex simplicial polyhedron 
whose vertices are ``almost'' on a sphere and whose skeleton has 
unbounded stretch factor. For simplicity of notation, we consider the 
sphere 
\[ \IS_3^2 = \{ (x,y,z) \in \IR^3 : x^2 + y^2 + z^2 = 3 \} .   
\] 
Let $k$ be a large integer and let $S_k$ be the subset of $\IR^3$ 
consisting of the following $12$ points:
\begin{itemize}
\item The $8$ vertices of the cube $[-1,1]^3$,
\item $p = (-1/k,0,a)$, where $a=\sqrt{3-1/k^2}$,
\item $q = (1/k,0,a)$,
\item $r = (0,-b,c)$, where $b=1/k^2$ and $c=a-1/k^3$, and 
\item $s = (0,b,c)$.
\end{itemize}
The $8$ vertices of the cube and the points $p$ and $q$ are on the
sphere $\IS_3^2$. Since  
\begin{eqnarray*} 
 b^2 + c^2 & = & 1/k^4 + a^2 - 2a/k^3 + 1/k^6 \\ 
    & < & 1/k^4 + a^2 + 1/k^6 \\ 
    & < & a^2 + 2/k^6 \\ 
 & = & 3 - 1/k^2 + 2/k^4 \\ 
 & < & 3 ,  
\end{eqnarray*}
the points $r$ and $s$ are in the interior of, but very close to, this 
sphere.

Let $P_k$ be the convex hull of the point set $S_k$. Below, we will 
show that, for sufficiently large values of $k$, 
(i) $(p,q,r)$ and $(p,q,s)$ are faces of the polyhedron $P_k$ and 
(ii) $rs$ is not an edge of $P_k$. Thus, the figure below shows (part of) 
the top view (in the negative $z$-direction) of $P_k$.

\begin{center}
   \includegraphics[scale=0.5]{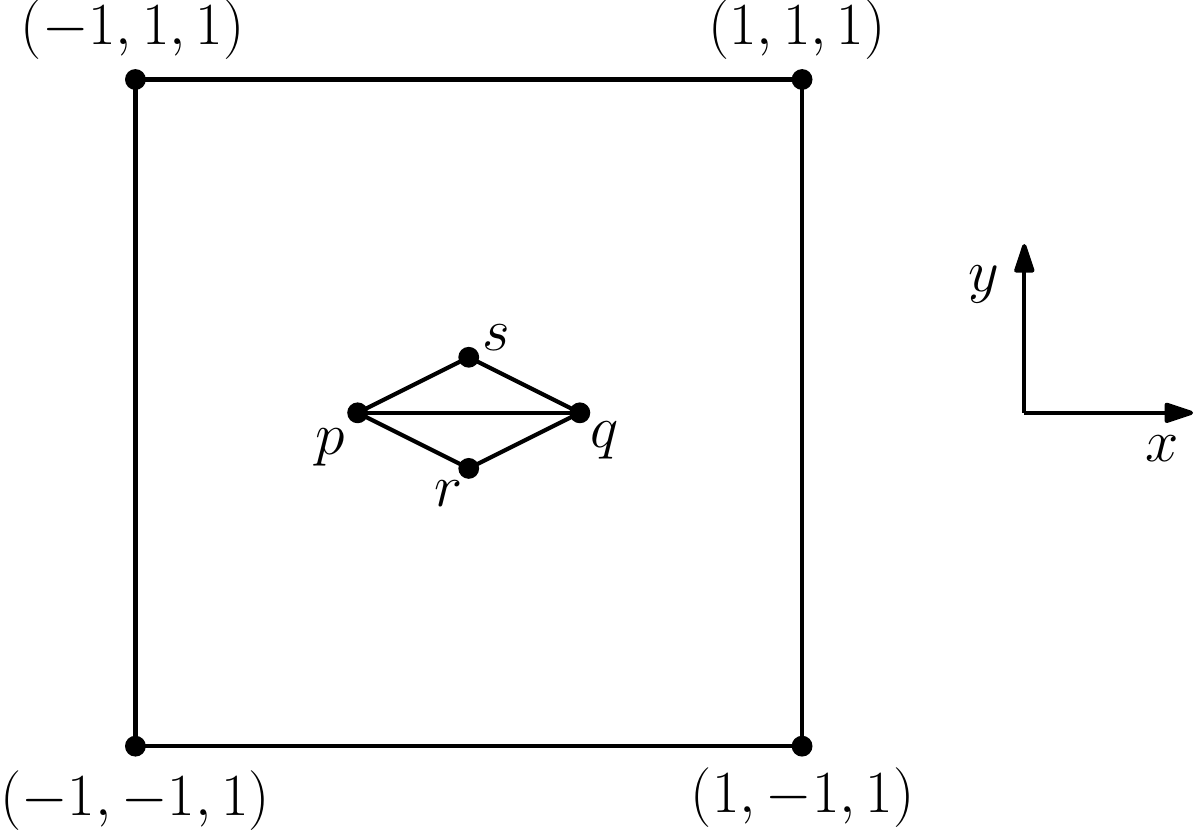}
\end{center}

The shortest path between $r$ and $s$ in the skeleton of $P_k$ has 
length
\[ |rp|+|ps| = 2|rp| , 
\]
which is at least twice the distance between $r$ and $p$ in the
$x$-direction, which is $2/k$. Since $|rs|=2b=2/k^2$, it follows that
the stretch factor of the skeleton of $P_k$ is at least
\[ \frac{2/k}{2/k^2} = k . 
\]
Thus, by letting $k$ go to infinity, the stretch factor of $\skel(P_k)$ 
is unbounded.

It remains to prove that $(p,q,r)$ and $(p,q,s)$ are faces of $P_k$ and  
$rs$ is not an edge of $P_k$. The plane through $p$, $q$, and $s$ has 
equation
\[ z = a - \frac{a-c}{b} \cdot y . 
\]
To prove that $(p,q,s)$ is a face of $P_k$, it suffices to show that the 
points $r$ and $(1,1,1)$ are below this plane. The point $r$ is below 
this plane if and only if 
\[ a - \frac{a-c}{b} \cdot (-b) > c , 
\]
which is equivalent to $a>c$, which obviously holds. The point $(1,1,1)$ 
is below this plane if and only if 
\[ a - \frac{a-c}{b} > 1 ,
\]
which is equivalent to
\begin{equation}  \label{eqtoprove} 
 ab -a + c > b . 
\end{equation}
Using the fact that, for sufficiently large values of $k$, $a > 3/2$,
we have
\begin{eqnarray*} 
  ab-a+c & = & a/k^2 - 1/k^3 \\ 
         & > & 3/(2k^2) - 1/k^3 \\ 
         & > & 1/k^2 \\ 
         & = & b ,  
\end{eqnarray*}
proving the inequality in (\ref{eqtoprove}). Thus, $(p,q,s)$ is a face 
of $P_k$. By a symmetric argument, $(p,q,r)$ is a face of $P_k$ as well. 

We finally show that $(r,s)$ is not an edge of $P_k$. For sufficiently 
large values of $k$, both points $r$ and $s$ are above (with respect to 
the $z$-direction) the points $(-1,1,1)$ and $(1,1,1)$.
Observe that both $r$ and $s$ are below $p$ and $q$.
It follows that for any plane through $r$ and $s$, (i) $(-1,1,1)$ and
$q$ are on opposite sides or (ii) $(1,1,1)$ and $p$ are on opposite sides.
Therefore, $(r,s)$ is not an edge of $P_k$.

Let $\alpha_k$ be the smallest angle in the face $(p,q,s)$ of the 
polyhedron $P_k$, i.e., $\alpha_k = \angle(qps)$. Then 
\[ (q-p) \cdot (s-p) = |pq| \, |ps| \cos \alpha_k , 
\] 
where $\cdot$ denotes the dot-product. A straightforward calculation 
shows that 
\[ \cos \alpha_k = 1 - \Theta(1/k^2) , 
\]
implying that $\alpha_k$ is proportional to $1/k$. Thus, as $k$ tends to 
infinity, the smallest angle in any face of $P_k$ tends to zero.  

Observe that the polyhedron $P_k$ does not satisfy the assumptions in 
Theorem~\ref{thm999}. By a sufficiently small perturbation of the points 
of $S_k$, however, we obtain a polyhedron that does satisfy these 
assumptions and whose skeleton has unbounded stretch factor. We 
conclude that Theorem~\ref{thm999} does not hold for all convex 
simplicial polyhedra whose vertices are very close to a sphere.

\section{Convex Cycles in an Annulus}  \label{secannulus} 
Let $r$ and $R$ be real numbers with $R>r>0$. Define $\Ann_{r,R}$ to be 
the \emph{annulus} consisting of all points in $\IR^2$ that are on or 
between the two circles of radii $r$ and $R$ that are centered at the 
origin. Thus,  
\[ \Ann_{r,R} = \{ (x,y) \in \IR^2 : r \leq x^2 + y^2 \leq R^2 \} . 
\]  
We will refer to the circles of radii $r$ and $R$ that are centered at 
the origin as the \emph{inner circle} and the \emph{outer circle} of the 
annulus, respectively. 

In this section, we consider convex polygons $Q$ that contain the origin 
and whose boundary is in $\Ann_{r,R}$. The skeleton $\skel(Q)$ of such 
a polygon is the graph whose vertex and edge sets are the vertex 
and edge sets of $Q$, respectively. 

Throughout this section, we will use the function $f$ defined by 
\[ f(x) = \sqrt{x^2-1} + x \cdot \arcsin(1/x)   
\]
for $x \geq 1$.  

\begin{lemma}    \label{lemfincr}
The function $f$ is increasing for $x \geq 1$. 
\end{lemma} 
\begin{proof}
The derivative of $f$ is given by 
\[ f'(x) = \frac{x-1}{\sqrt{x^2-1}} + \arcsin(1/x) . 
\] 
It is clear that $f'(x)>0$ for $x>1$. 
\end{proof}  

We will prove the following result:  
 
\begin{theorem}    \label{thmpolygon}  
Let $r$ and $R$ be real numbers with $R>r>0$ and let $Q$ be a convex 
polygon that contains the origin in its interior and whose boundary is 
contained in the annulus $\Ann_{r,R}$. Then the skeleton of $Q$ is an 
$f(R/r)$-spanner of the vertex set of $Q$. 
\end{theorem} 

Theorem~\ref{thmpolygon} refers to the stretch factor of $\skel(Q)$, 
which is the maximum value of $|pq|_{\skel(Q)} / |pq|$ over all
pairs of distinct vertices $p$ and $q$ of $Q$. It turns out that the 
proof becomes simpler if we also consider points that are in the interior 
of edges. This gives rise to the notion of geometric dilation, which 
we recall in the following subsection.

\subsection{Geometric Dilation of Convex Cycles} 
Let $C$ be a convex cycle in $\IR^2$. We observe that $C$ is 
rectifiable, i.e., its length, denoted by $|C|$, is well-defined; 
see, for example, Section~1.5 in Toponogov~\cite{t-dgcs-06}. For any two 
distinct points $p$ and $q$ on $C$, there are two paths along $C$ that 
connect $p$ and $q$. We denote the length of the shorter of these two 
paths by $|pq|_C$. The \emph{geometric dilation} of $C$ is defined as 
\[ \Dil(C) = \max_{p,q \in C, p \neq q} \frac{|pq|_C}{|pq|} .
\]  
Ebbers-Baumann \emph{et al.}~\cite{egk-gdcpc-07} have proved that, for 
a convex cycle $C$, $\Dil(C)$ is well-defined. That is, the maximum 
in the definition of $\Dil(C)$ exists.   

Let $p$ and $q$ be two points on $C$. We say that these two points form a 
\emph{halving pair} if the two paths along $C$ between $p$ and $q$ have 
the same length. 

\begin{lemma}[Ebbers-Baumann \emph{et al.}~\cite{egk-gdcpc-07}] 
\label{lemEB}  
Let $C$ be a convex cycle in $\IR^2$, and let $h$ be the minimum 
Euclidean distance between the points of any halving pair. Then the 
geometric dilation of $C$ is attained by a halving pair with Euclidean 
distance $h$ and   
\[ \Dil(C) = \frac{|C|/2}{h} . 
\] 
\end{lemma} 

\subsection{Convex Cycles in an Annulus}  \label{secCC} 
In this section, we consider convex cycles $C$ that contain the origin 
in their interior and that are contained in the annulus $\Ann_{r,R}$. 
We will prove that the geometric dilation of such a cycle is at most 
$f(R/r)$, where $f$ is the function defined in the beginning of 
Section~\ref{secannulus}. Clearly, this result will imply  
Theorem~\ref{thmpolygon}.  

We start by giving an example of a convex cycle whose geometric dilation 
is equal to $f(R/r)$. Let $C^*$ be the convex cycle that consists of 
the two vertical tangents at the inner circle of $\Ann_{r,R}$ that have 
their endpoints at the outer circle, and the two arcs on the outer 
circle that connect these tangents; see the figure below. 

\begin{center}
   \includegraphics[scale=0.6]{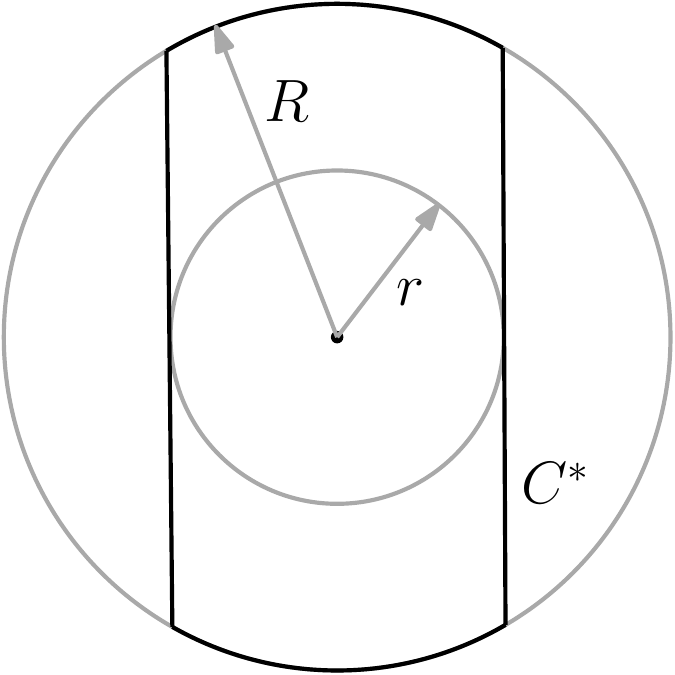}
\end{center}

A simple calculation shows that the length of $C^*$ satisfies 
\[ | C^* | = 4 \sqrt{R^2-r^2} + 4R \cdot \arcsin(r/R) = 
           4r \cdot f(R/r) . 
\] 

\begin{lemma}  \label{lemCstar} 
The geometric dilation of $C^*$ satisfies 
\[ \Dil \left( C^* \right) = f(R/r) . 
\] 
\end{lemma} 
\begin{proof} 
Consider any halving pair $p,q$ of $C^*$. Since $C^*$ is centrally 
symmetric with respect to the origin, we have $q=-p$. The inner circle
of $\Ann_{r,R}$ is between the two lines through $p$ and $q$ that are 
orthogonal to the line segment $pq$. Therefore, $|pq| \geq 2r$. 
Thus, by Lemma~\ref{lemEB}, 
\[ \Dil \left( C^* \right) \leq \frac{|C^*|}{4r} = f(R/r) .  
\] 
If we take for $p$ and $q$ the leftmost and rightmost points of the inner 
circle, then $|pq|_{C^*} / |pq| = f(R/r)$. Therefore, we have 
$\Dil \left( C^* \right) = f(R/r)$. 
\end{proof}  

In the following lemmas, we consider special types of convex cycles in 
$\Ann_{r,R}$. For each such type, we prove an upper bound of $f(R/r)$ 
on their geometric dilation. In Theorem~\ref{thmcycle}, we will consider 
the general case and reduce the problem of bounding the geometric 
dilation to one of the special types.   

\begin{lemma}   \label{lemoutercircle} 
Let $C$ be a convex cycle in $\Ann_{r,R}$ that contains the origin in 
its interior, and let $p$ and $q$ be two distinct points on $C$ 
such that $\Dil(C) = |pq|_C / |pq|$. If both $p$ and $q$ are on the 
outer circle of $\Ann_{r,R}$, then $\Dil(C) \leq f(R/r)$. 
\end{lemma} 
\begin{proof} 
Let $C'$ denote the outer circle of $\Ann_{r,R}$. Then   
\[ \Dil(C) = \frac{|pq|_C}{|pq|} \leq \frac{|pq|_{C'}}{|pq|} \leq  
   \Dil \left( C' \right) = \pi/2 = f(1) . 
\] 
Since, by Lemma~\ref{lemfincr}, $f(1) \leq f(R/r)$, it follows that 
$\Dil(C) \leq f(R/r)$. 
\end{proof} 

\begin{lemma}   \label{lemthree}  
Consider a line segment $L$ that is tangent to the inner circle of 
$\Ann_{r,R}$ and has both endpoints on the outer circle. Let $C$ be the 
convex cycle that consists of $L$ and the longer arc on the outer circle 
that connects the endpoints of $L$. Then 
\[ \Dil(C) \leq f(R/r) .
\]  
\end{lemma} 
\begin{proof} 
We may assume without loss of generality that $L$ is horizontal and 
touches the lowest point of the inner circle; see the figure below. 

\begin{center}
   \includegraphics[scale=0.6]{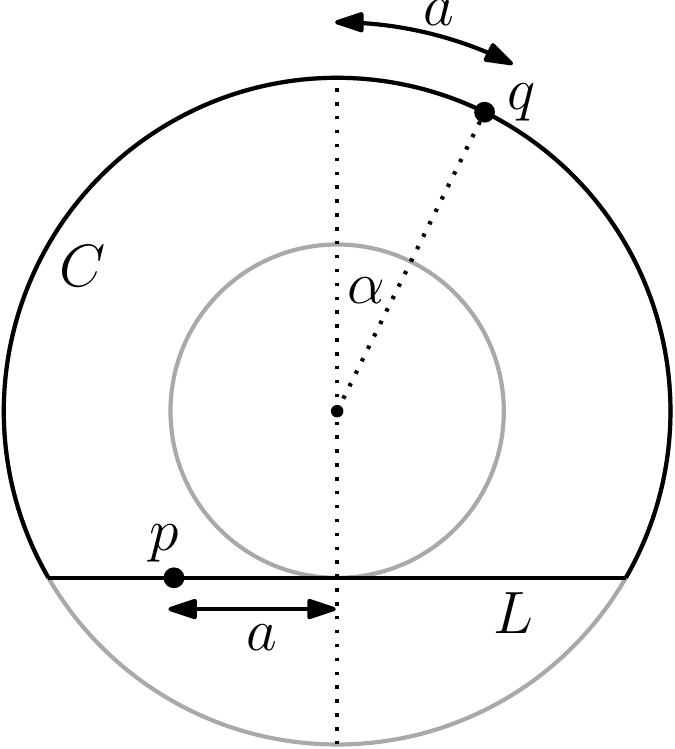}
\end{center}

Let $p$ and $q$ form a halving pair of $C$ that attains the geometric 
dilation of $C$. Observe that at least one of $p$ and $q$ is on the 
outer circle of $\Ann_{r,R}$. If both $p$ and $q$ are on the outer 
circle, then $\Dil(C) \leq f(R/r)$ by Lemma~\ref{lemoutercircle}. 
Otherwise, we may assume without loss of generality that 
(i) $p$ is on $L$ and on or to the left of the $y$-axis, and 
(ii) $q$ is on the outer circle, on or to the right of the $y$-axis 
and above the $x$-axis. 

We first prove that $|pq| \geq R+r$. Let $p$ have coordinates $p=(-a,-r)$ 
for some real number $a$ with $0 \leq a \leq \sqrt{R^2-r^2}$, and 
let $\alpha$ be the angle between the $y$-axis and the vector from the 
origin to $q$; see the figure above. Since $p$ and $q$ form a halving 
pair, the clockwise arc from the highest point on the outer circle to 
the point $q$ has length $a$. Therefore, $\alpha = a/R$ and, thus, 
the coordinates of the point $q$ are 
$q = (R \cdot \sin(a/R),R \cdot \cos(a/R))$.   
If we define the function $g$ by 
\[ g(a) = (a+R \cdot \sin(a/R))^2 + (r+R \cdot \cos(a/R))^2  
\]
for $0 \leq a \leq \sqrt{R^2-r^2}$, then $|pq|^2 = g(a)$. 
The derivative of $g$ satisfies 
\[ g'(a) = 2a + 2(R-r) \cdot \sin(a/R) + 2a \cdot \cos(a/R) , 
\]
which is positive for $0 < a \leq \sqrt{R^2-r^2}$. Therefore, the 
function $g$ is increasing and 
\[ |pq|^2 = g(a) \geq g(0) = (R+r)^2 , 
\]
implying that $|pq| \geq R+r$. 

We conclude that 
\[ \Dil(C) = \frac{|pq|_C}{|pq|} = \frac{|C|/2}{|pq|} \leq 
           \frac{|C|}{2(R+r)} . 
\]  
To complete the proof, it suffices to show that 
\begin{equation}  \label{eq666} 
         \frac{|C|}{2(R+r)} \leq f(R/r) . 
\end{equation}  
We observe that the length of $C$ satisfies 
\[ |C| = 2 \sqrt{R^2-r^2} + 2R \cdot \arcsin (r/R) + \pi R . 
\] 
Recall that 
\[ f(x) = \sqrt{x^2-1} + x \cdot \arcsin(1/x) .   
\]
Thus, (\ref{eq666}) becomes 
\[ \frac{\sqrt{R^2-r^2} + R \cdot \arcsin (r/R) + \pi R/2}{R+r} 
    \leq 
   \frac{\sqrt{R^2-r^2} + R \cdot \arcsin (r/R)}{r} .  
\]
The latter inequality is equivalent to 
\[ \pi/2 \leq \frac{\sqrt{R^2-r^2} + R \cdot \arcsin (r/R)}{r} , 
\] 
i.e., 
\[ f(1) \leq f(R/r) . 
\] 
Since the latter inequality follows from Lemma~\ref{lemfincr}, 
we have shown that (\ref{eq666}) holds. 
\end{proof}  

\begin{lemma}   \label{lemfour}  
Let $b$ be a point in $\Ann_{r,R}$ that is on the negative $y$-axis. 
Let $a$ and $c$ be two points on the outer circle of $\Ann_{r,R}$ 
such that (i) both $a$ and $c$ have the same $y$-coordinate and are
below the $x$-axis and (ii) both line segments $ab$ and $bc$ are 
tangent to the inner circle of $\Ann_{r,R}$. 
Let $C$ be the convex cycle that consists of the two line segments $ab$
and $bc$, and the longer arc on the outer circle that connects $a$ 
and $c$. Then 
\[ \Dil(C) \leq f(R/r) .
\]  
\end{lemma} 
\begin{proof}  
The figure below illustrates the situation. 

\begin{center}
   \includegraphics[scale=0.6]{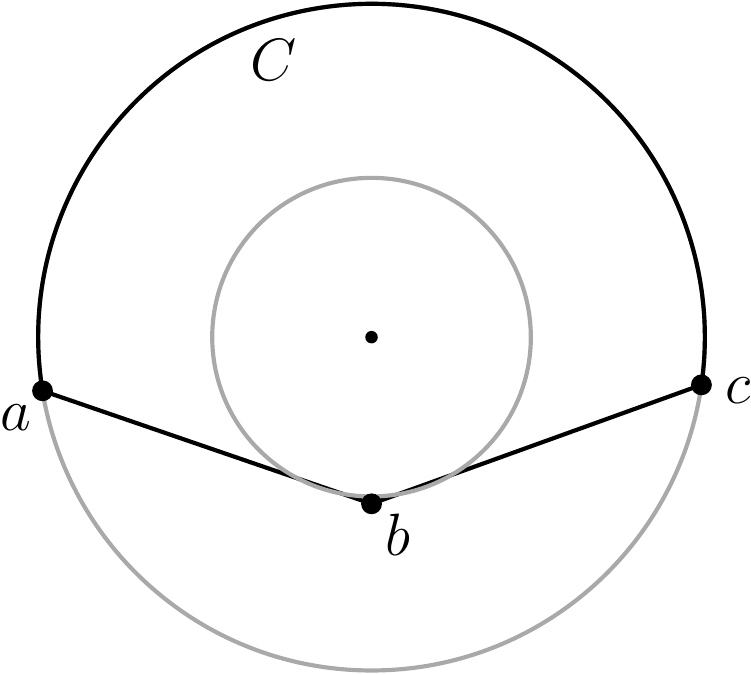}
\end{center}

Let $p$ and $q$ form a halving pair of $C$ that attains the geometric 
dilation of $C$. Observe that at least one of $p$ and $q$ is on the 
outer circle of $\Ann_{r,R}$. If both $p$ and $q$ are on the outer 
circle of $\Ann_{r,R}$, then $\Dil(C) \leq f(R/r)$ by 
Lemma~\ref{lemoutercircle}. Otherwise, we may assume without loss of 
generality that $p$ is on the outer circle of $\Ann_{r,R}$ and $q$ is 
on the line segment $bc$, as in the figure below. 

\begin{center}
   \includegraphics[scale=0.6]{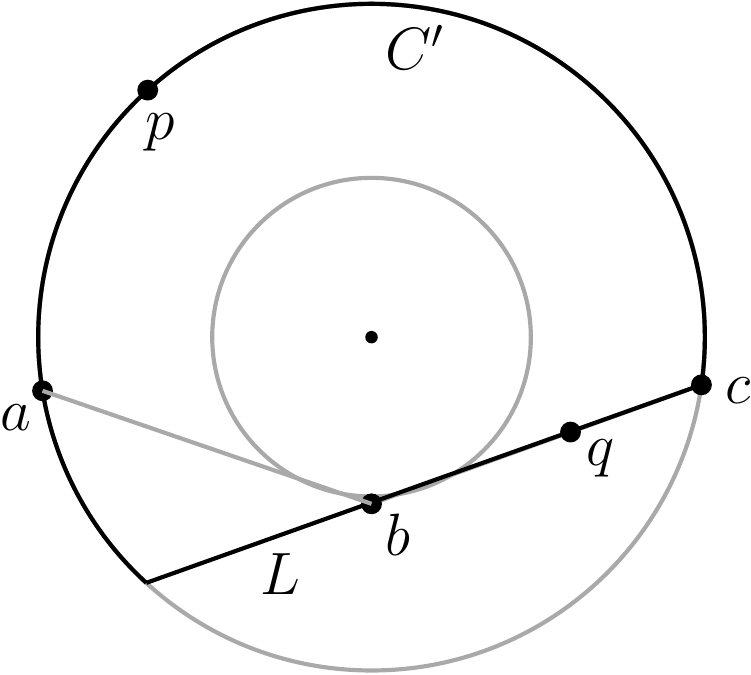}
\end{center}

Let $L$ be the maximal line segment in $\Ann_{r,R}$ that contains the 
segment $bc$. Let $C'$ be the convex cycle consisting of $L$ and the 
longer arc on the outer circle connecting the two endpoints of $L$. Then 
\[ \Dil(C) = \frac{|pq|_C}{|pq|} \leq \frac{|pq|_{C'}}{|pq|} \leq  
   \Dil \left( C' \right) . 
\] 
By Lemma~\ref{lemthree}, we have $\Dil \left( C' \right) \leq f(R/r)$. 
It follows that $\Dil(C) \leq f(R/r)$. 
\end{proof}  

\begin{lemma}   \label{lemtwo}  
Consider two non-crossing line segments $L_1$ and $L_2$ that are tangent 
to the inner circle of $\Ann_{r,R}$ and have their endpoints on the outer 
circle. Let $C$ be the convex cycle that consists of $L_1$, $L_2$, 
and the two arcs on the outer circle that connect $L_1$ and $L_2$; 
one of these two arcs may consist of a single point. Then 
\[ \Dil(C) \leq f(R/r) .
\]  
\end{lemma} 
\begin{proof}  
We may assume without loss of generality that $C$ is symmetric with 
respect to the $y$-axis and $L_1$ is to the left of $L_2$, as in the 
figure below. 

\begin{center}
   \includegraphics[scale=0.6]{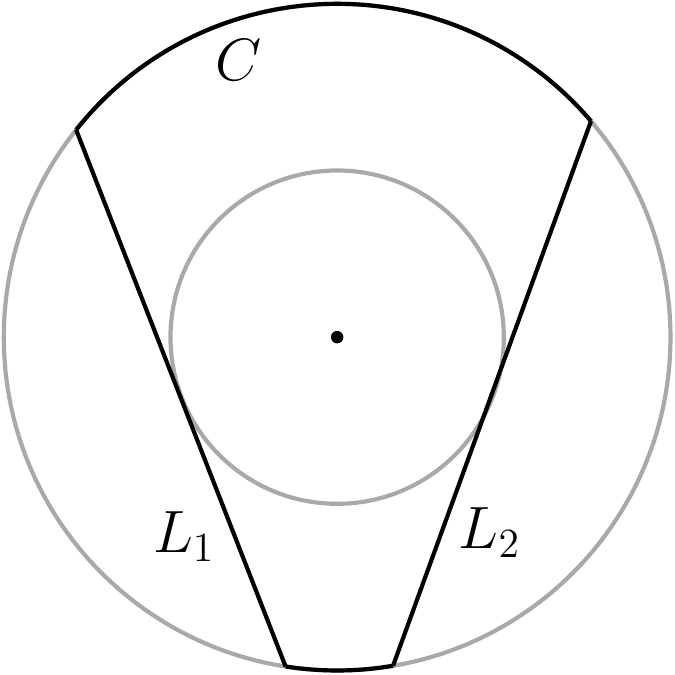}
\end{center}

If $L_1$ and $L_2$ are parallel, then the claim follows from 
Lemma~\ref{lemCstar}. Thus, we assume that $L_1$ and $L_2$ are not 
parallel. We may assume without loss of generality that the length of 
the lower arc of $C$ is less than the length of the upper arc, as 
in the figure above. We observe that 
\[ | C | = 4r \cdot f(R/r) , 
\] 
i.e., the length of $C$ is equal to the length of the cycle $C^*$ in 
Lemma~\ref{lemCstar}. Indeed, if we rotate $L_2$, while keeping it 
tangent to the inner circle, until it becomes parallel to $L_1$, then 
the length of the cycle does not change. 

Let $p$ and $q$ form a halving pair of $C$ that attains the geometric 
dilation of $C$, i.e., 
\[ \Dil(C) = \frac{|pq|_C}{|pq|} = \frac{|C|/2}{|pq|} . 
\]  
We consider three cases for the locations of $p$ and $q$ on $C$. 

\vspace{0.5em} 

\noindent 
{\bf Case 1:} Both $p$ and $q$ are on the outer circle of $\Ann_{r,R}$. 

Then we have $\Dil(C) \leq f(R/r)$ by Lemma~\ref{lemoutercircle}. 

\vspace{0.5em} 

\noindent 
{\bf Case 2:} $p$ is on the outer circle of $\Ann_{r,R}$ and $q$ is not 
on the outer circle. 

Since $p$ and $q$ form a halving pair, $p$ must be on the upper arc of 
$C$. We may assume without loss of generality that $q$ is on $L_2$.  
Let $C'$ be the convex cycle consisting of $L_2$ and the longer arc 
on the outer circle connecting the two endpoints of $L_2$; see the 
figure below.  
\begin{center}
   \includegraphics[scale=0.6]{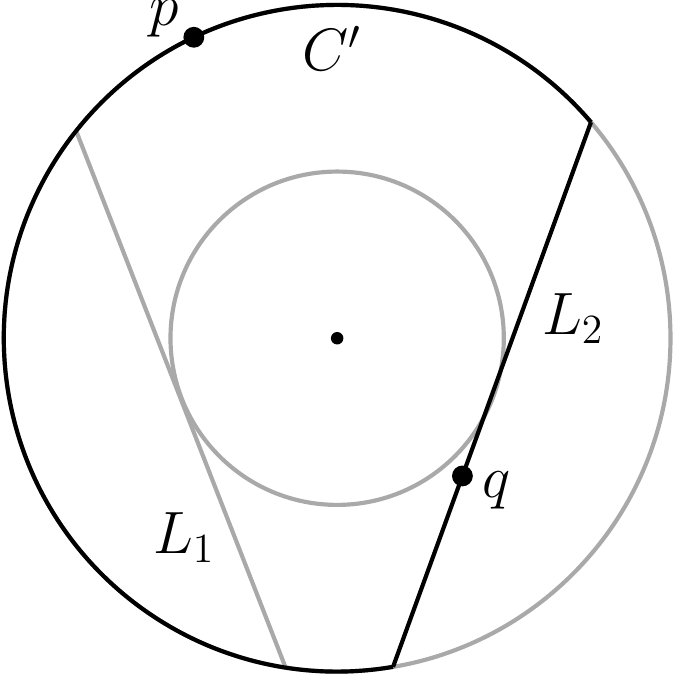}
\end{center}

We have 
\[ \Dil(C) = \frac{|pq|_C}{|pq|} \leq \frac{|pq|_{C'}}{|pq|} \leq  
   \Dil \left( C' \right) . 
\] 
By Lemma~\ref{lemthree}, we have $\Dil \left( C' \right) \leq f(R/r)$. 
It follows that $\Dil(C) \leq f(R/r)$. 

\vspace{0.5em} 

\noindent 
{\bf Case 3:} Neither $p$ nor $q$ is on the outer circle of $\Ann_{r,R}$. 

Since $p$ and $q$ form a halving pair, these two points cannot both be 
on the same line segment of $C$. We may assume without loss of 
generality that $p$ is on $L_1$ and $q$ is on $L_2$. 

Let $L'_1$ be the maximal line segment in $\Ann_{r,R}$ that is parallel 
and not equal to $L_1$ and that touches the inner circle.  
Let $L'_2$ be the maximal line segment in $\Ann_{r,R}$ that is parallel 
and not equal to $L_2$ and that touches the inner circle.  

\begin{center}
   \includegraphics[scale=0.6]{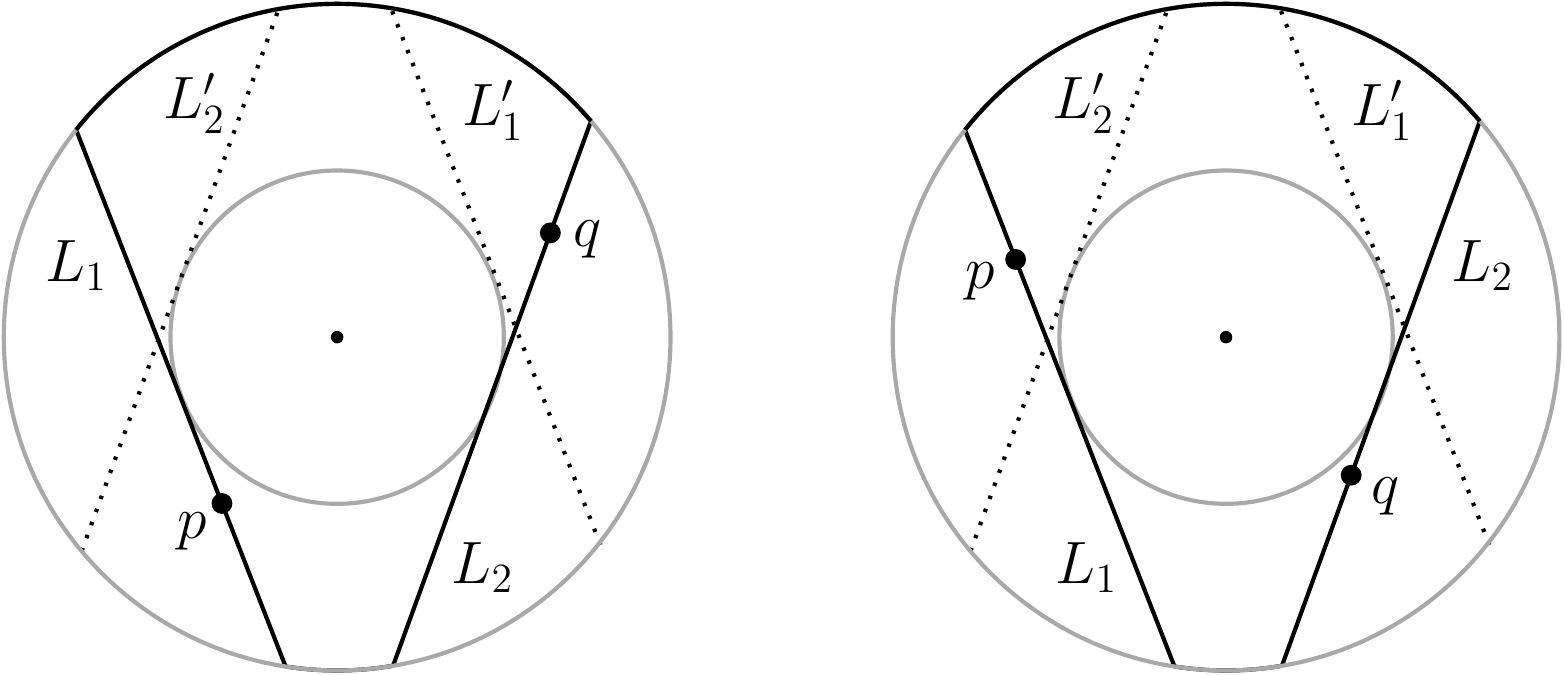}
\end{center}

We claim that $q$ is to the right of $L'_1$ or $p$ is to the left of 
$L'_2$; see the two figures above. Assuming this is true, it follows 
that $|pq| \geq 2r$ and 
\[ \Dil(C) = \frac{|C|/2}{|pq|} \leq \frac{|C|}{4r} = f(R/r) . 
\]  
To prove the claim, assume that $q$ is to the left of $L'_1$ and $p$ is 
to the right of $L'_2$. Let $a$ be the intersection between $L_1$ and 
$L'_2$, and let $b$ be the intersection between $L'_1$ and $L_2$; see 
the figure below. 

\begin{center}
   \includegraphics[scale=0.6]{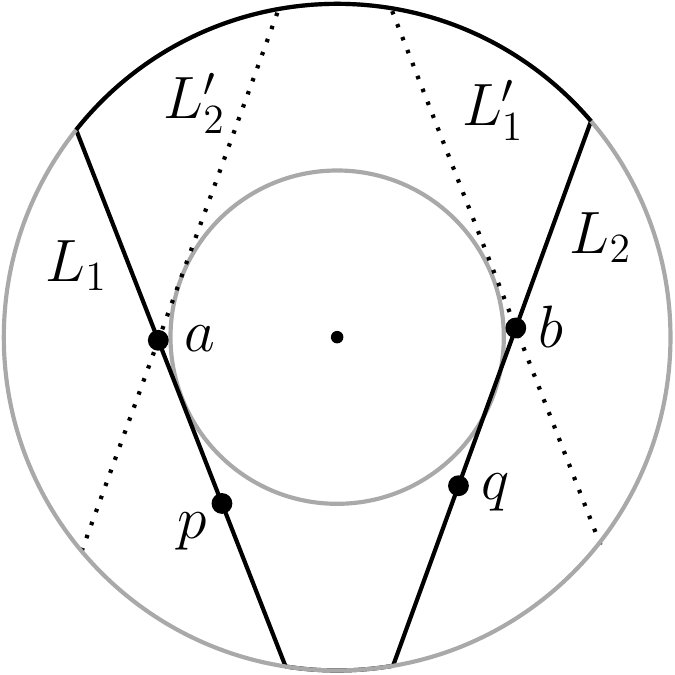}
\end{center}

Observe that both $a$ and $b$ are on the $x$-axis, and both $p$ and $q$ 
are below the $x$-axis. Therefore, the part of $C$ below the $x$-axis is 
shorter than the part of $C$ above the $x$-axis. Thus, the two paths 
along $C$ between $p$ and $q$ do not have the same lengths. This  
contradicts our assumption that $p$ and $q$ form a halving pair of 
$C$.     
\end{proof}  

We are now ready to consider an arbitrary convex cycle $C$ that contains 
the origin in its interior and that is contained in $\Ann_{r,R}$. 
A \emph{homothet} of $C$ is obtained by scaling $C$ with respect to the 
origin, followed by a translation. Observe that the dilation of a 
homothet of $C$ is equal to the dilation of $C$. 

\begin{theorem}    \label{thmcycle}  
Let $r$ and $R$ be real numbers with $R>r>0$ and let $C$ be a convex 
cycle that contains the origin in its interior and that is contained in 
the annulus $\Ann_{r,R}$. Then  
\[ \Dil(C) \leq f(R/r) , 
\] 
where $f$ is the function defined in the beginning of 
Section~\ref{secannulus}. 
\end{theorem} 
\begin{proof} 
Let $p$ and $q$ form a halving pair of $C$ that attains the geometric 
dilation of $C$, i.e., 
\[ \Dil(C) = \frac{|pq|_C}{|pq|} .  
\]  
We first assume that neither $p$ nor $q$ is on the outer circle of 
$\Ann_{r,R}$. 
 
Let $L_p$ and $L_q$ be supporting lines of $C$ through $p$ and $q$, 
respectively. Since $p$ and $q$ form a halving pair, $L_p \neq L_q$.  
Let $C_1$ be the convex cycle of maximum length in $\Ann_{r,R}$ that is 
between $L_p$ and $L_q$. Observe that $C_1$ contains two line segments 
such that (i) all their four endpoints are on the outer circle (as in 
the left figure below) or (ii) two of their endpoints are on the outer 
circle, whereas the other two endpoints meet in the interior of 
$\Ann_{r,R}$ (as in the right figure below). If (i) holds, we say that 
$C_1$ is of \emph{type~1}. In the other case, i.e., if (ii) holds, we 
say that $C_1$ is of \emph{type~2}. We have 
\[ \Dil(C) = \frac{|pq|_C}{|pq|} \leq \frac{|pq|_{C_1}}{|pq|} . 
\]  

\begin{center}
   \includegraphics[scale=0.6]{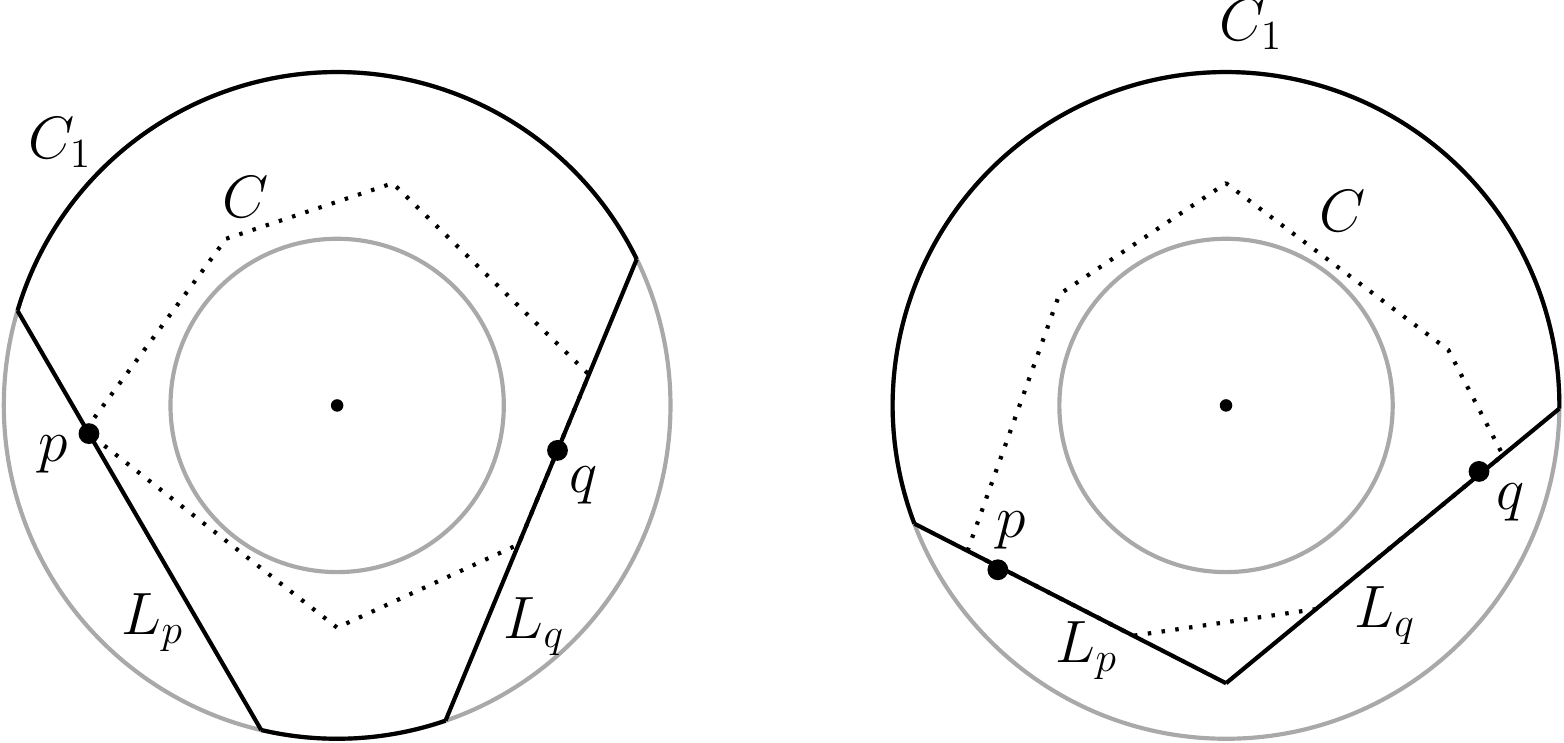}
\end{center}

We claim that there is a homothet $C_2$ of $C_1$ that is contained in 
$\Ann_{r,R}$ and that touches the inner circle in two points; see the 
two figures below. 

\begin{center}
   \includegraphics[scale=0.6]{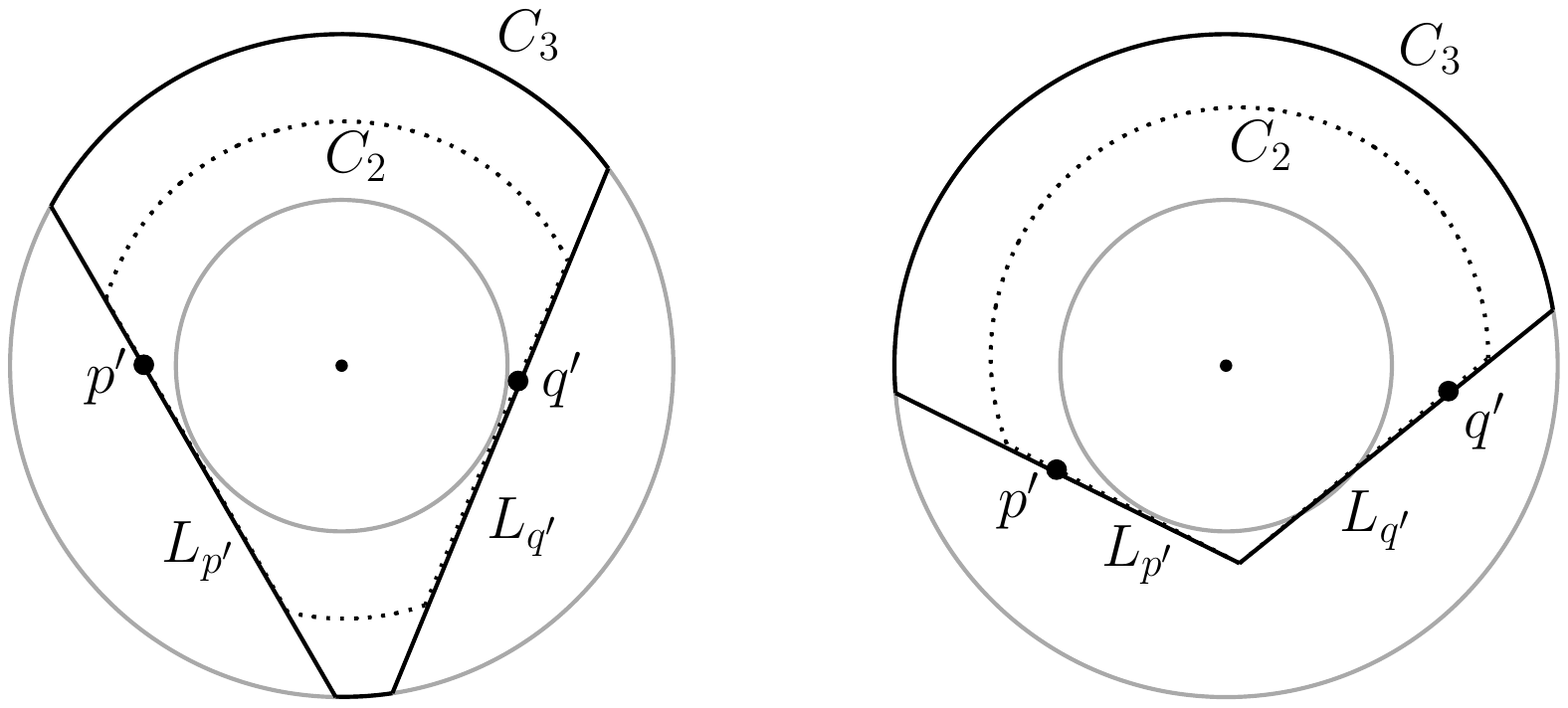}
\end{center}

To obtain such a homothet $C_2$, we do the following. First, we shrink 
$C_1$, i.e., scale it (with respect to the origin) by a factor of less 
than one, until it touches the inner circle. At this moment, one 
of the lines $L_p$ and $L_q$ in the shrunken copy of $C_1$ touches 
the inner circle. Assume, without loss of generality, that $L_q$ 
touches the inner circle, whereas $L_p$ does not. Let $c$ denote 
the ``center'' of the scaled copy of $C_1$, which is the origin. 
We translate $c$ towards $L_q$ in the direction that is orthogonal to 
$L_q$. During this translation, we shrink $C_1$ (with respect to its 
center $c$) while keeping $L_q$ on its boundary. We stop translating 
$c$ as soon as $L_p$ touches the inner circle of $\Ann_{r,R}$.
The resulting translated and shrunken copy of $C_1$ is the homothet 
$C_2$. 

Let $p'$ and $q'$ be the two points on the homothet $C_2$ that correspond 
to $p$ and $q$, respectively. Then  
\[ \Dil(C) \leq \frac{|pq|_{C_1}}{|pq|} = \frac{|p'q'|_{C_2}}{|p'q'|} . 
\]  
Let $L_{p'}$ and $L_{q'}$ be supporting lines of $C_2$ through $p'$ 
and $q'$, respectively, and let $C_3$ be the convex cycle of maximum 
length in $\Ann_{r,R}$ that is between $L_{p'}$ and $L_{q'}$; see the 
two figures above. Observe that $C_3$ is either of type~1 or of type~2. 
In fact, $C_3$ may be of type~2, even if $C_1$ is of type~1. We have 
\[ \Dil(C) \leq \frac{|p'q'|_{C_2}}{|p'q'|} \leq 
         \frac{|p'q'|_{C_3}}{|p'q'|} . 
\] 
First assume that $C_3$ is of type~1. Thus, all four endpoints of the 
two line segments of $C_3$ are on the outer circle of $\Ann_{r,R}$ (as 
in the left figure above). Then $C_3$ satisfies the conditions of 
Lemma~\ref{lemtwo} and, therefore, 
\[ \Dil(C) \leq \frac{|p'q'|_{C_3}}{|p'q'|}  \leq 
           \Dil \left( C_3 \right) \leq f(R/r) . 
\]  

Now assume that $C_3$ is of type~2. We may assume without loss of 
generality that $C_3$ is symmetric with respect to the $y$-axis, 
and the intersection point of $L_{p'}$ and $L_{q'}$ is on the 
negative $y$-axis. Translate $C_3$ in the negative $y$-direction 
until it touches the outer circle. Denote the resulting translate 
by $C_4$. Let $p''$ and $q''$ be the two points on $C_4$ that correspond 
to $p'$ and $q'$, respectively. Then  
\[ \Dil(C) \leq \frac{|p'q'|_{C_3}}{|p'q'|} = 
          \frac{|p''q''|_{C_4}}{|p''q''|} . 
\] 
We consider two cases. 

\vspace{0.5em} 

\noindent 
{\bf Case 1:} The lowest point of $C_4$ is on the outer circle of 
$\Ann_{r,R}$; see the left figure below. 

\begin{center}
   \includegraphics[scale=0.6]{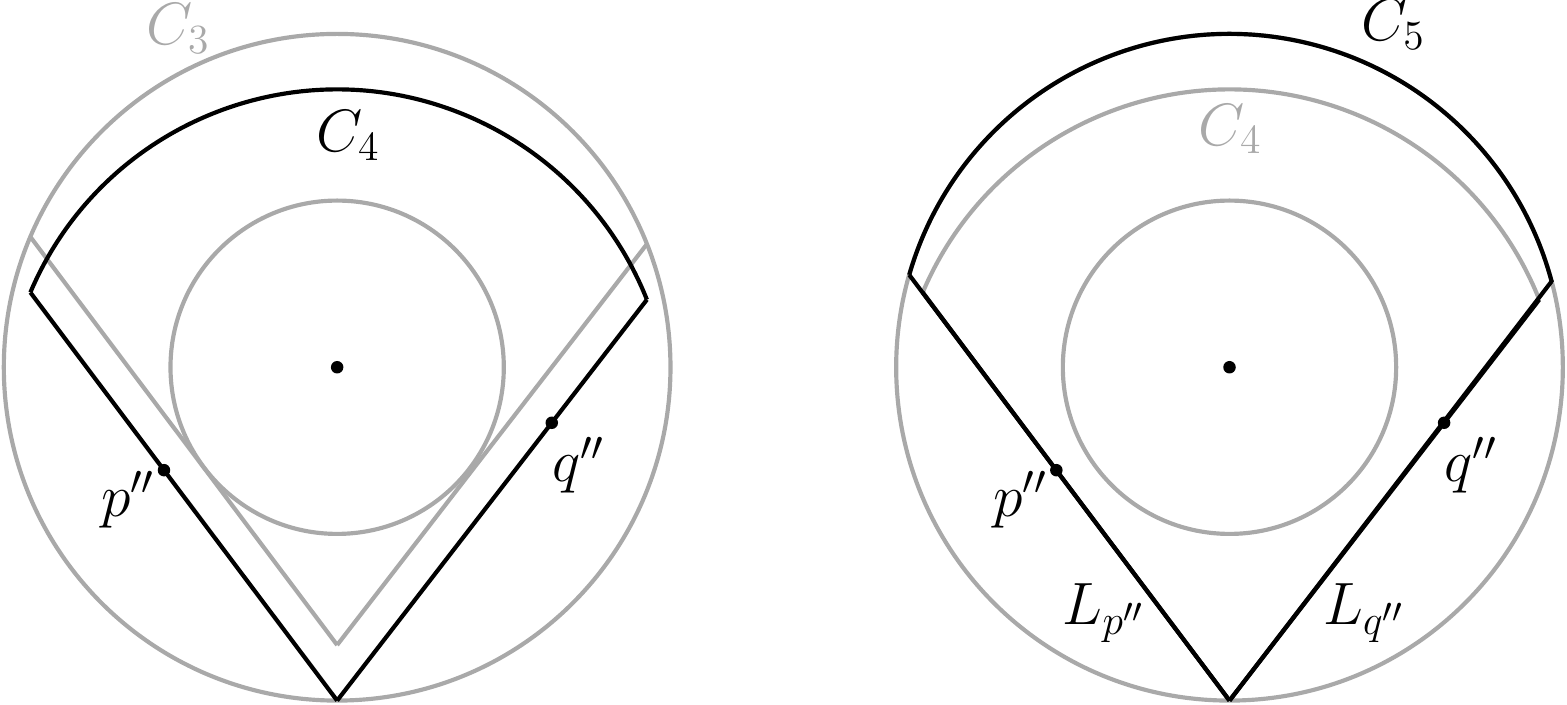}
\end{center}

Let $L_{p''}$ and $L_{q''}$ be supporting lines of $C_4$ through $p''$ 
and $q''$, respectively, and let $C_5$ be the convex cycle of maximum 
length in $\Ann_{r,R}$ that is between $L_{p''}$ and $L_{q''}$; see 
the right figure above. Observe that 
\[ \Dil(C) \leq \frac{|p''q''|_{C_4}}{|p''q''|} \leq  
     \frac{|p''q''|_{C_5}}{|p''q''|} \leq \Dil \left( C_5 \right) . 
\] 
Enlarge the inner circle of $\Ann_{r,R}$ such that it touches the 
two line segments of $C_5$. Denoting the radius of this enlarged 
circle by $r'$, it follows from Lemmas~\ref{lemtwo} and~\ref{lemfincr} 
that 
\[ \Dil(C) \leq \Dil \left( C_5 \right) \leq f(R/r') \leq f(R/r) . 
\] 

\vspace{0.5em} 

\noindent 
{\bf Case 2:} The leftmost and rightmost points of $C_4$ are on the 
outer circle of $\Ann_{r,R}$; see the left figure below. 

\begin{center}
   \includegraphics[scale=0.6]{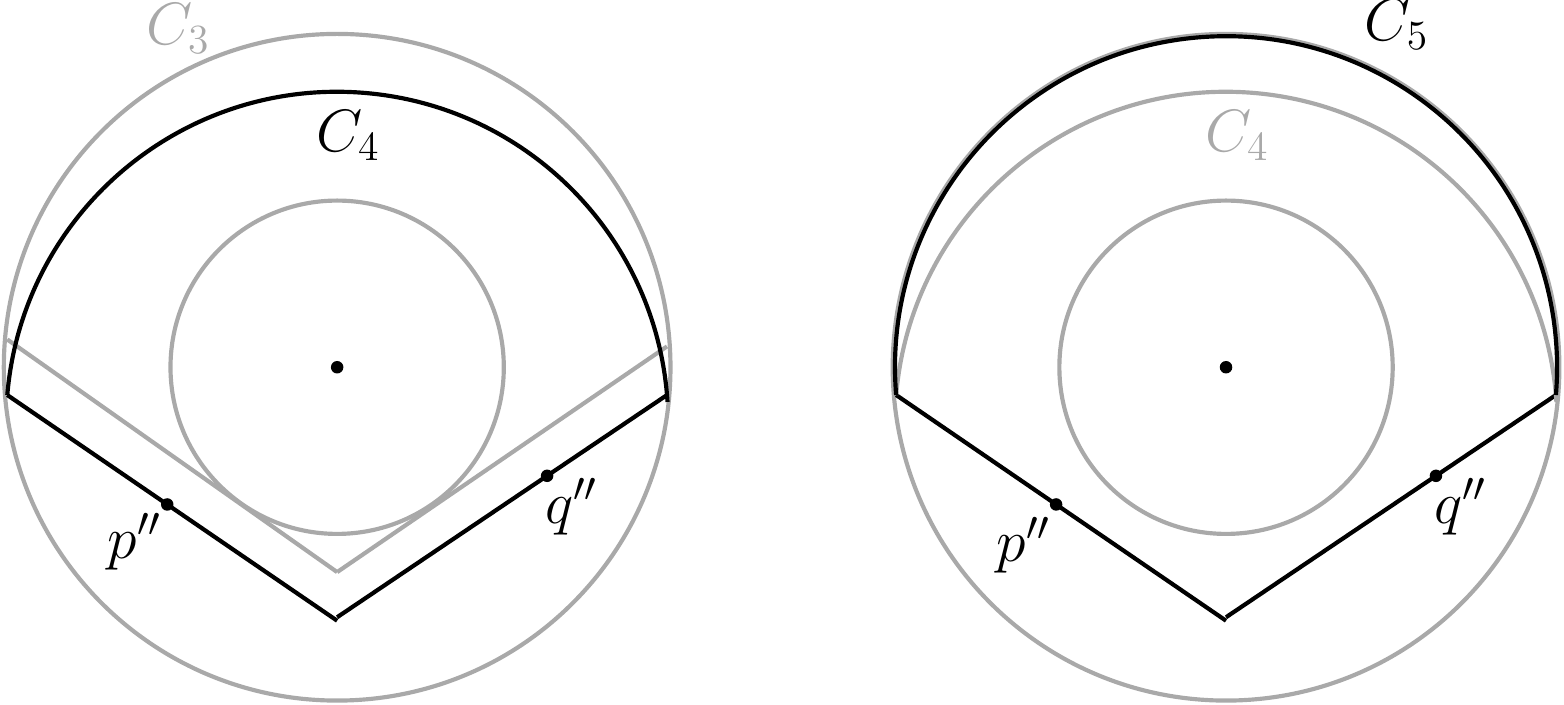}
\end{center}

Let $C_5$ be the convex cycle consisting of the two line segments 
of $C_4$ and the upper arc on the outer circle connecting them; see 
the right figure above. Then 
\[ \Dil(C) \leq \frac{|p''q''|_{C_4}}{|p''q''|} \leq  
     \frac{|p''q''|_{C_5}}{|p''q''|} \leq \Dil \left( C_5 \right) . 
\] 
Enlarge the inner circle of $\Ann_{r,R}$ such that it touches the 
two line segments of $C_5$. Let $r'$ be the radius of this enlarged 
circle. Since $C_5$ satisfies the conditions of Lemma~\ref{lemfour}
for $\Ann_{r',R}$, we have  
\[ \Dil \left( C_5 \right) \leq f(R/r') \leq f(R/r) . 
\] 
Thus, we have shown that $\Dil(C) \leq f(R/r)$. 

\vspace{0.5em} 

Until now we have assumed that neither $p$ nor $q$ is on the outer 
circle of $\Ann_{r,R}$. Assume now that $p$ or $q$ is on this outer 
circle. Let $\varepsilon>0$ be an arbitrary real number and consider 
the annulus $\Ann_{r,R+\varepsilon}$. Since neither $p$ nor $q$ is on 
the outer circle of this enlarged annulus, the analysis given above 
implies that 
\[ \Dil(C) \leq f((R+\varepsilon)/r) . 
\] 
Thus, since this holds for any $\varepsilon>0$, we have 
\[ \Dil(C) \leq \inf_{\varepsilon>0} f((R+\varepsilon)/r) . 
\] 
Since the function $f$ is continuous, it follows from 
Lemma~\ref{lemfincr} that 
\[ \Dil(C) \leq \inf_{\varepsilon>0} f((R+\varepsilon)/r) = f(R/r) . 
\] 
This concludes the proof. 
\end{proof}

\section{Angle-Constrained Convex Polyhedra in a Spherical Shell} 
\label{secACSH} 

Let $r$ and $R$ be real numbers with $R>r>0$. Define $\Shell_{r,R}$ 
to be the \emph{spherical shell} consisting of all points in $\IR^3$ 
that are on or between the two spheres of radii $r$ and $R$ that are 
centered at the origin. In other words, 
\[ \Shell_{r,R} = \{ (x,y,z) \in \IR^3 : 
            r \leq x^2 + y^2 + z^2 \leq R^2 \} . 
\] 

In this section, we consider convex simplicial polyhedra that contain 
the origin in their interiors and whose boundaries are contained in 
$\Shell_{r,R}$. From Section~\ref{secalmost}, the skeletons of such 
polyhedra can have unbounded stretch factors. 

Let $\theta$ be a real number with $0 < \theta < \pi/3$. We say that a 
convex polyhedron $P$ is $\theta$-\emph{angle-constrained}, if the 
angles in all faces of $P$ are at least $\theta$. 

Let $P$ be a convex simplicial polyhedron that contains the origin in 
its interior, whose boundary is contained in $\Shell_{r,R}$, and that is 
$\theta$-angle-constrained. In this section, we prove that the stretch 
factor of the skeleton of $P$ is bounded from above by a function of 
$R/r$ and $\theta$. Our proof will use an improvement of a result by 
Karavelas and Guibas~\cite{kg-skgsa-01} about chains of triangles in 
$\IR^2$; see Lemma~\ref{lemchain}. We start by reviewing such chains.

\subsection{Chains of Triangles}  \label{secCoT} 
Before we define chains of triangles, we prove a geometric lemma that 
will be used later in this section. 

\begin{lemma}  \label{lemtriangle}  
Let $a$, $b$, and $c$ be three pairwise distinct points in the plane, 
and let $\alpha = \angle(bac)$. Then 
\[ |ab| + |ac| \leq \frac{|bc|}{\sin (\alpha/2)} .
\] 
\end{lemma} 
\begin{proof} 
Consider the interior angle bisector of $a$; see the figure below. 

\begin{center}
\includegraphics[scale=0.7]{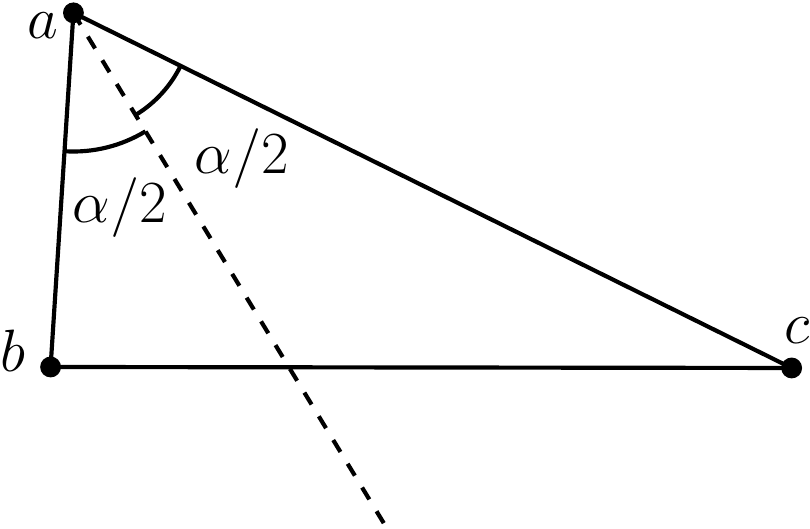}
\end{center}

Let $\ell_b$ be the distance between $b$ and this bisector, and let 
$\ell_c$ be the distance between $c$ and this bisector. Then 
\[ |ab| + |ac| = \frac{\ell_b}{\sin(\alpha/2)} + 
                 \frac{\ell_c}{\sin(\alpha/2)} 
  \leq \frac{|bc|}{\sin(\alpha/2)} . 
\]
\end{proof} 

Let $p$ and $q$ be two distinct points in $\IR^2$, let $k \geq 2$ be 
an integer, and consider a sequence 
$\mathcal{T} = (T_1,T_2,\ldots,T_k)$ of triangles in $\IR^2$. The 
sequence $\mathcal{T}$ is called a 
\emph{chain of triangles with respect to $p$ and $q$}, if 
\begin{enumerate} 
\item $p$ is a vertex of $T_1$, but not of $T_2$,  
\item $q$ is a vertex of $T_k$, but not of $T_{k-1}$,  
\item for each $i$ with $1 \leq i < k$, the interiors of the triangles 
      $T_i$ and $T_{i+1}$ are disjoint and these triangles share an edge,
      and 
\item for each $i$ with $1 \leq i \leq k$, the line segment $pq$ 
      intersects the interior of $T_i$. 
\end{enumerate} 
See Figure~\ref{figchain} for examples. 

\begin{figure}
\begin{center}
\includegraphics[scale=0.7]{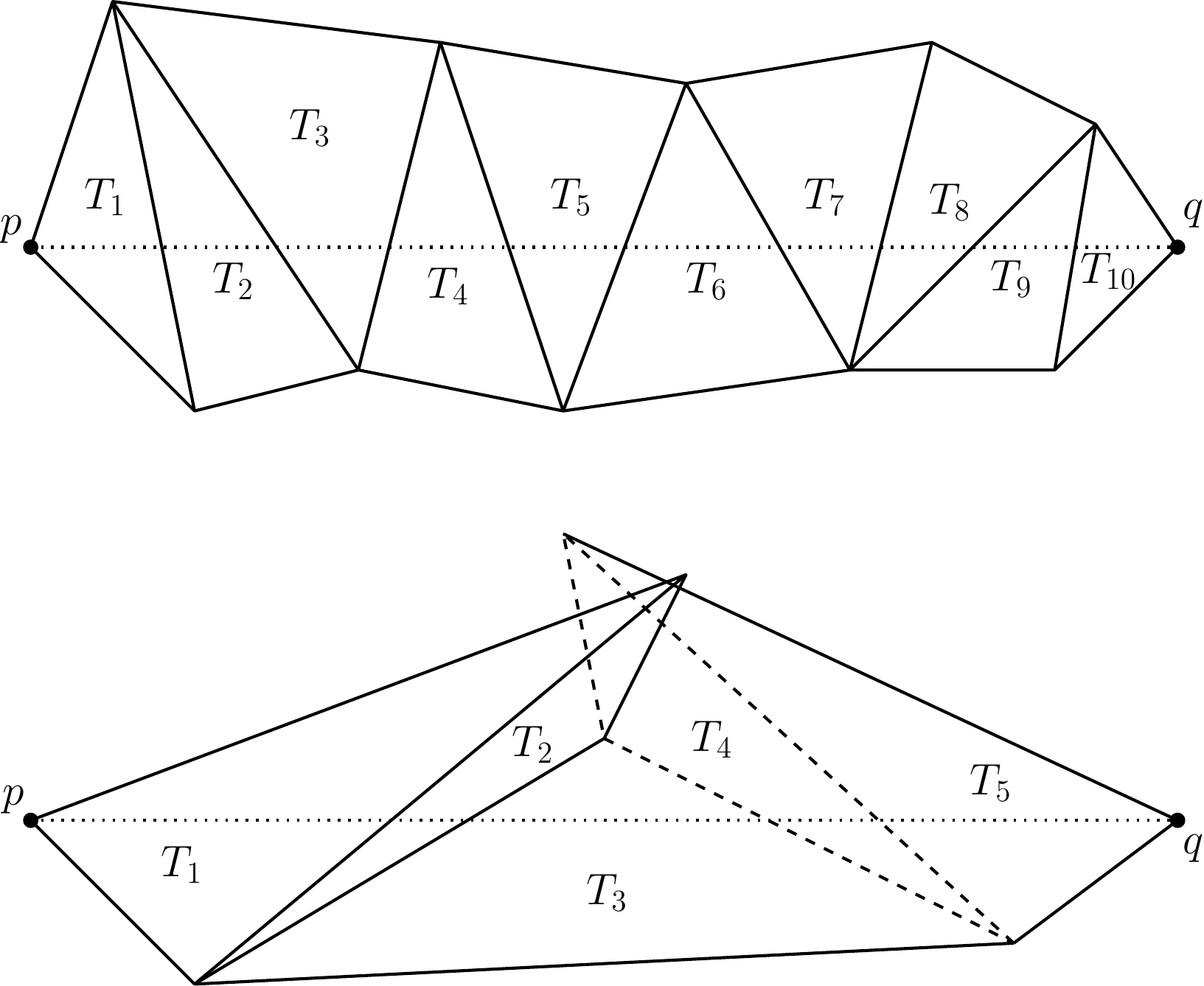}
\end{center}
\caption{Two examples of chains of triangles with respect to the points 
         $p$ and $q$. For clarity, the triangle $T_4$ in the second 
         example is dashed.} 
\label{figchain}
\end{figure}

Let $G(\mathcal{T})$ be the graph whose vertex and edge sets consist of 
all vertices and edges of the $k$ triangles in $\mathcal{T}$, 
respectively. The length of each edge in this graph is equal to the 
Euclidean distance between its vertices. The length of a shortest path 
in $G(\mathcal{T})$ is denoted by $|pq|_{G(\mathcal{T})}$. 

\begin{lemma}
\label{lemchain}  
Let $\theta$ be a real number with $0 < \theta < \pi/3$, let $p$ and $q$ 
be two distinct points in the plane, and let $\mathcal{T}$ be a 
chain of triangles with respect to $p$ and $q$. Assume that all angles 
in any of the triangles in $\mathcal{T}$ are at least $\theta$. Then 
\[ |pq|_{G(\mathcal{T})} \leq 
      \frac{1+ 1 / \sin(\theta/2)}{2} \cdot |pq| .
\]
\end{lemma} 
\begin{proof}
We assume, without loss of generality, that the line segment $pq$ is 
on the $x$-axis and $p$ is to the left of $q$. We start by constructing 
a preliminary path in $G(\mathcal{T})$ from $p$ to $q$ (this is the 
same path as in Karavelas and Guibas~\cite{kg-skgsa-01}): 
\begin{enumerate} 
\item Let $pr$ be one of the two edges of the triangle $T_1$ with 
      endpoint $p$. We initialize the path to be $(p,r)$. 
\item Consider the current path and let $r$ be its last point. 
      Assume that $r \neq q$. 
      \begin{enumerate} 
      \item If $r$ is below the $x$-axis, then consider all edges in  
            $G(\mathcal{T})$ that have $r$ as an endpoint and whose other 
            endpoint is on or above the $x$-axis. Let $rr'$ be the 
            ``rightmost'' of these edges, i.e., the edge among these 
            whose angle with the positive $x$-axis is minimum. Then we 
            extend the path by the edge $rr'$, i.e., we add the point 
            $r'$ at the end of the current path.
      \item If $r$ is above the $x$-axis, then consider all edges in 
            $G(\mathcal{T})$ that have $r$ as an endpoint and whose other 
            endpoint is on or below the $x$-axis. Let $rr'$ be the 
            ``rightmost'' of these edges, i.e., the edge among these 
            whose angle with the positive $x$-axis is maximum. Then we 
            extend the path by the edge $rr'$, i.e., we add the point 
            $r'$ at the end of the current path.
      \end{enumerate} 
\end{enumerate} 

\begin{figure}
\begin{center}
\includegraphics[scale=0.7]{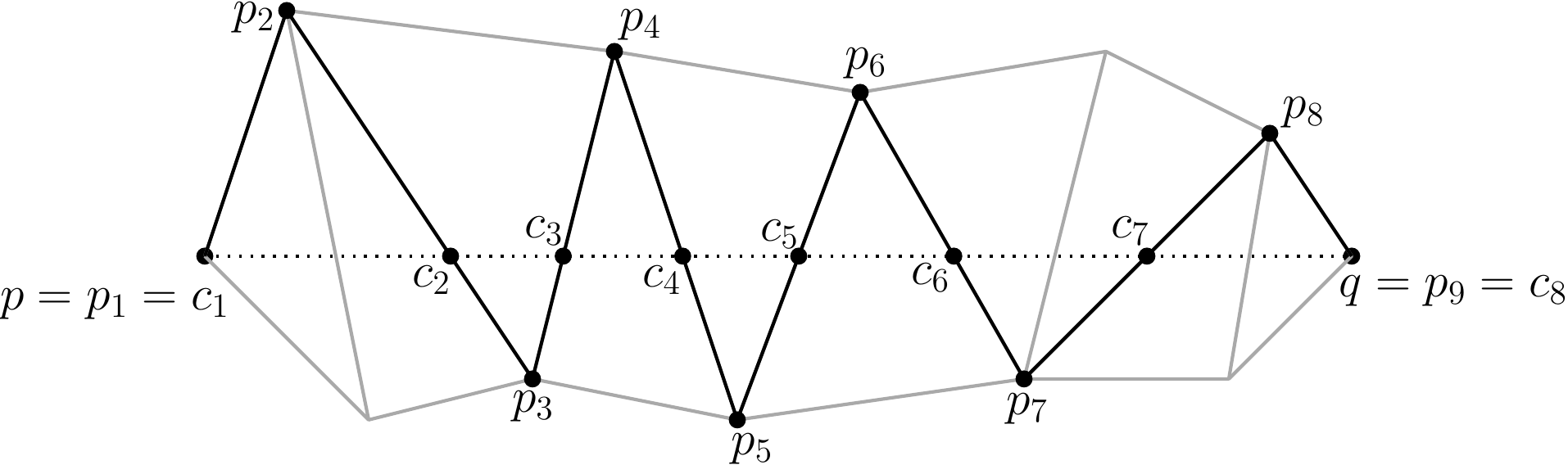}
\end{center}
\caption{The path $\Pi = (p=p_1,p_2,\ldots,p_9=q)$ in the first chain of 
         triangles in Figure~\ref{figchain}. The segments $c_1 c_2$, 
         $c_6 c_7$, and $c_7 c_8$ belong to group~1, the segments 
         $c_3 c_4$ and $c_5 c_6$ belong to group~2, and the segments 
         $c_2 c_3$ and $c_4 c_5$ belong to group~3. The path $\Pi'$ 
         is equal to $(p=p_1,p_2,p_3,p_5,p_7,p_8,p_9=q)$.}
\label{figchainpath}
\end{figure}

Number the triangles in $\mathcal{T}$ as $T_1,T_2,\ldots,T_k$, in the 
order in which they are intersected by the line segment from $p$ to 
$q$. Then the point $r'$ is a vertex of a triangle in $\mathcal{T}$ 
that has a larger index than the index of any triangle that contains 
the vertex $r$. Therefore, if we continue extending the path, it will 
reach the point $q$. Denote the resulting path by 
$\Pi = (p=p_1,p_2,\ldots,p_{\ell}=q)$; see Figure~\ref{figchainpath}. 

As a warming-up, we prove an upper bound on the length of the path $\Pi$.
For each $i$ with $1 \leq i < \ell$, let $c_i$ be the intersection 
between the line segments $pq$ and $p_i p_{i+1}$. Then 
$|pq|_{G(\mathcal{T})}$ is at most the length of the path $\Pi$, i.e., 
\[ |pq|_{G(\mathcal{T})} \leq 
     \sum_{i=1}^{\ell-1} |p_i p_{i+1}| =  
     \sum_{i=1}^{\ell-1} \left( |c_i p_{i+1}| + |p_{i+1} c_{i+1}| 
                         \right) .
\]  
Let $\alpha_i = \angle(c_i p_{i+1} c_{i+1})$. Since 
$\alpha_i \geq \theta$, it follows from Lemma~\ref{lemtriangle} that 
\[ |c_i p_{i+1}| + |p_{i+1} c_{i+1}| \leq 
         \frac{|c_i c_{i+1}|}{\sin (\alpha_i/2)} \leq  
         \frac{|c_i c_{i+1}|}{\sin (\theta/2)} .
\]  
Therefore, 
\[ |pq|_{G(\mathcal{T})} \leq 
     \sum_{i=1}^{\ell-1} \frac{|c_i c_{i+1}|}{\sin (\theta/2)} =  
     \frac{|pq|}{\sin (\theta/2)} . 
\]   

To improve the upper bound on $|pq|_{G(\mathcal{T})}$, we divide the 
line segments $c_i c_{i+1}$, $1 \leq i < \ell$, into three groups: A 
segment $c_i c_{i+1}$ belongs to \emph{group~1} if its 
relative interior intersects an edge of some triangle of the chain 
$\mathcal{T}$. If the relative interior of $c_i c_{i+1}$ is entirely 
contained in one of the triangles of $\mathcal{T}$ and the point 
$p_{i+1}$ is on or above the $x$-axis, then $c_i c_{i+1}$ belongs to 
\emph{group~2}. Otherwise, $c_i c_{i+1}$ belongs to \emph{group~3}.  
Refer to Figure~\ref{figchainpath} for an illustration.  
For $j=1,2,3$, let $X_j$ denote the total length of all line segments 
$c_i c_{i+1}$ in group~$j$. We may assume without loss of generality 
that $X_3 \leq X_2$. 

Consider again the path $\Pi = (p=p_1,p_2,\ldots,p_{\ell}=q)$. 
For each $i$ such that $c_i c_{i+1}$ belongs to group~2, we replace 
the subpath $(p_i,p_{i+1},p_{i+2})$ in $\Pi$ by the short-cut 
$p_i p_{i+2}$; refer to Figure~\ref{figchainpath}. Let $\Pi'$ denote 
the resulting path from $p$ to $q$. 

For each line segment $c_i c_{i+1}$ in group~1, we have 
$\alpha_i \geq 2 \theta$. If $c_i c_{i+1}$ is in group~2, then 
\[ |p_i p_{i+2} | \leq |p_i c_i| + |c_i c_{i+1}| + |c_{i+1} p_{i+2}| . 
\] 
It follows that 
\[ |pq|_{G(\mathcal{T})} \leq | \Pi' | \leq  
     \frac{X_1}{\sin\theta} + X_2 + \frac{X_3}{\sin(\theta/2)} .
\] 
Recall that $0 \leq X_3 \leq X_2$. This inequality is equivalent to 
\begin{eqnarray*} 
  X_2 + \frac{X_3}{\sin(\theta/2)} & \leq & 
        \frac{1}{2} 
        \left( 1 + \frac{1}{\sin(\theta/2)} \right) 
        \left( X_2+X_3 \right) \\ 
        & = & \frac{1}{2} 
              \left( 1 + \frac{1}{\sin(\theta/2)} \right) 
              \left( |pq| - X_1 \right) 
\end{eqnarray*} 
implying that 
\begin{equation}  \label{pqG} 
 |pq|_{G(\mathcal{T})} \leq 
    \left( \frac{1}{\sin\theta} - \frac{1}{2} - 
           \frac{1}{2\sin(\theta/2)} 
    \right) X_1 + 
        \frac{1}{2} \left( 1 + \frac{1}{\sin(\theta/2)} \right) |pq| . 
\end{equation} 
The function $g(\theta) = 1/\sin\theta - 1/2 - 1/(2\sin(\theta/2))$ 
is negative for $0 < \theta < \pi/3$. To prove this, using 
$\sin\theta = 2 \sin(\theta/2) \cos(\theta/2)$ and a straightforward 
calculation, we observe that $g(\theta) < 0$ if and only if 
\[ \frac{\sin\theta}{2} + \cos(\theta/2) > 1 . 
\] 
The left-hand side in the above inequality has a positive derivative 
for $0 < \theta < \pi/3$ (this can be verified using 
$\cos\theta = 1 - 2 \sin^2(\theta/2)$); thus, 
\[ \frac{\sin\theta}{2} + \cos(\theta/2) >  
   \frac{\sin 0}{2} + \cos(0/2) = 1 . 
\] 

We conclude that, since $g(\theta)<0$, the first term on the right-hand 
side in (\ref{pqG}) is non-positive, implying that 
\[ |pq|_{G(\mathcal{T})} \leq 
     \frac{1+1/\sin(\theta/2)}{2} \cdot |pq| . 
\] 
This completes the proof.  
\end{proof}

\subsection{Angle-Constrained Convex Polyhedra}   
Let $\theta$ be a real number with $0 < \theta < \pi/3$ and let $P$ be 
a convex simplicial polyhedron that is $\theta$-angle-constrained. In 
this section, we bound the ratio of the shortest-path distance 
$|pq|_{\skel(P)}$ between $p$ and $q$ in the skeleton of $P$ and the 
shortest-path distance $|pq|_{\partial P}$ between $p$ and $q$ along 
the surface of $P$.  

Let $p$ and $q$ be two distinct vertices of $P$ and consider the 
shortest path $\Pi_{pq}$ along the surface of $P$ from $p$ to $q$. 
Except for $p$ and $q$, this path does not contain any vertex of $P$; 
see Sharir and Schorr~\cite{ss-spps-86}. Let $T_1,T_2,\ldots,T_k$ be 
the sequence of faces of $P$ that this path passes through. Let 
$\mathcal{T'} = (T'_1,T'_2,\ldots,T'_k)$ be the sequence of triangles 
obtained from an edge-unfolding of the triangles $T_1,T_2,\ldots,T_k$. 
Let $p'$ and $q'$ be the vertices of $T'_1$ and $T'_k$ corresponding to 
$p$ and $q$, respectively. Sharir and Schorr~\cite{ss-spps-86} 
(see also Agarwal \emph{et al.}~\cite{aaos-supa-97}) have shown that 
\begin{itemize}
\item $\mathcal{T'}$ is a chain of triangles with respect to $p'$ and 
      $q'$, as defined in Section~\ref{secCoT}, and 
\item the path $\Pi_{pq}$ along $\partial P$ unfolds to the line segment 
      $p'q'$, i.e., $|pq|_{\partial P} = |\Pi_{pq}| = |p'q'|$. 
\end{itemize} 
Consider the graph $G(\mathcal{T'})$ that is defined by the chain 
$\mathcal{T'}$; see Section~\ref{secCoT}. Observe that $|pq|_{\skel(P)}$ 
is at most the shortest-path distance between $p$ and $q$ in the graph 
consisting of all vertices and edges of the triangles 
$T_1,T_2,\ldots,T_k$. The latter shortest-path distance is equal to 
$|p'q'|_{G(\mathcal{T'})}$. Thus, using Lemma~\ref{lemchain}, we obtain 
\begin{eqnarray*} 
 |pq|_{\skel(P)} & \leq & |p'q'|_{G(\mathcal{T'})} \\ 
   & \leq & 
    \frac{1+1/\sin(\theta/2)}{2} \cdot |p'q'| \\  
     & = & 
      \frac{1+1/\sin(\theta/2)}{2} \cdot |pq|_{\partial P} .  
\end{eqnarray*} 
We have proved the following result: 

\begin{lemma}  \label{lempartial} 
Let $\theta$ be a real number with $0 < \theta < \pi/3$ and let $P$ be 
a $\theta$-angle-constrained convex simplicial polyhedron. 
For any two distinct vertices $p$ and $q$ of $P$, we have 
\[ |pq|_{\skel(P)} \leq 
           \frac{1+1/\sin(\theta/2)}{2} \cdot|pq|_{\partial P} .  
\]
\end{lemma} 

\subsection{Angle-Constrained Convex Polyhedra in a Spherical Shell}   
We are now ready to prove the main result of Section~\ref{secACSH}: 

\begin{theorem} 
Let $r$, $R$, and $\theta$ be real numbers with $R>r>0$ and 
$0 < \theta < \pi/3$, and let $P$ be a $\theta$-angle-constrained 
convex simplicial polyhedron that contains the origin and whose boundary 
is contained in the spherical shell $\Shell_{r,R}$. Then the skeleton of 
$P$ is a $t$-spanner of the vertex set of $P$, where  
\[ t = \frac{1+1/\sin(\theta/2)}{2} 
       \left( \sqrt{(R/r)^2-1} + (R/r) \arcsin(r/R) \right) . 
\] 
\end{theorem} 
\begin{proof}  
Let $p$ and $q$ be two distinct vertices of $P$. 
By Lemma~\ref{lempartial}, we have 
\[ |pq|_{\skel(P)} \leq 
           \frac{1+1/\sin(\theta/2)}{2} \cdot |pq|_{\partial P} .  
\]
Let $H_{pq}$ be the plane through $p$, $q$, and the origin, and let 
$Q_{pq}$ be the intersection of $P$ and $H_{pq}$. Then 
\[ |pq|_{\partial P} \leq |pq|_{\partial Q_{pq}} . 
\]  
Since $Q_{pq}$ is a convex polygon satisfying the conditions of 
Theorem~\ref{thmpolygon}, we have 
\[ |pq|_{\partial Q_{pq}} \leq 
        \left( \sqrt{(R/r)^2-1} + (R/r) \arcsin(r/R) \right) \cdot |pq| . 
\] 
\end{proof} 

\section{Concluding Remarks} 
We have considered the problem of bounding the stretch factor of the 
skeleton of a convex simplicial polyhedron $P$ in $\IR^3$. If the 
vertices of $P$ are on a sphere, then this stretch factor is at 
most $0.999 \cdot \pi$, which is $\pi/2$ times the currently best known 
upper bound on the stretch factor of the Delaunay triangulation in 
$\IR^2$. We obtained this result from Xia's upper bound on the stretch 
factor of chains of disks in~\cite{x-sfdtl-13}. Observe that Xia's 
result implies an upper bound on the stretch factor of the Delaunay 
triangulation. The converse, however, may not be true, because the 
chains of disks that arise in the analysis of the Delaunay triangulation 
are much more restricted than general chains of disks; see for example 
Figure~2 in~\cite{x-sfdtl-13}. Thus, an improved upper bound on the 
stretch factor of the Delaunay triangulation may not imply an improved 
upper bound on the stretch factor of the skeleton of $P$. Nevertheless, 
we make the following conjecture: Let $t^*$ be a real number such that 
the stretch factor of any Delaunay triangulation in $\IR^2$ is at most 
$t^*$. Then the stretch factor of the skeleton of any convex polyhedron 
in $\IR^3$, all of whose vertices are on a sphere, is at most 
$t^* \cdot \pi/2$. 
 
We have shown that the skeleton of a convex simplicial polyhedron $P$ 
whose vertices are ``almost'' on a sphere may have an unbounded 
stretch factor. For the case when $P$ contains the origin, its 
boundary is contained in the spherical shell $\Shell_{r,R}$, and the 
angles in all faces are at least $\theta$, we have shown that the stretch 
factor of $P$'s skeleton is bounded from above by a function that 
depends only on $R/r$ and $\theta$. We leave as an open problem to find 
other classes of convex polyhedra whose skeletons have bounded stretch 
factor.

\section*{Acknowledgments} 
Part of this work was done at the 
\emph{Third Annual Workshop on Geometry and Graphs}, held at the 
Bellairs Research Institute in Barbados, March 8--13, 2015.  
We thank the other workshop participants for their helpful comments. 

We thank the anonymous referees for their useful comments. 
We especially thank one of the referees for simplifying the proofs of 
Lemmas~\ref{lemlocallyD} and~\ref{lemtriangle}, and for suggesting 
the use of short-cuts in the proof of Lemma~\ref{lemchain}.

\end{document}